\documentclass[11pt]{article}
\usepackage[utf8]{inputenc}
\usepackage[margin=1in]{geometry}
\usepackage{amsthm}
\usepackage{thm-restate}
\newtheorem{theorem}{Theorem}
\newtheorem{lemma}[theorem]{Lemma}
\newtheorem{definition}[theorem]{Definition}
\newtheorem{problem}[theorem]{Problem}
\newtheorem{subproblem}[theorem]{Subproblem}
\newtheorem{observation}[theorem]{Observation}

\usepackage{amsmath}
\usepackage{amsfonts}
\def\frechet{d_F}
\def\graphdist{d_P}
\DeclareMathOperator{\diameter}{diameter}
\DeclareMathOperator{\radius}{radius}
\DeclareMathOperator{\trough}{trough}
\DeclareMathOperator{\poly}{poly}
\usepackage{xcolor}
\usepackage{enumitem}
\usepackage{graphicx}
\usepackage{authblk}
\usepackage{hyperref}

\title{Map matching queries on realistic input graphs under the Fr\'echet distance}

\author[1]{Joachim Gudmundsson}
\author[2]{Martin P. Seybold}
\author[3]{Sampson Wong}
\affil[1]{
     \small{University of Sydney, Australia.}
     \small{joachim.gudmundsson@sydney.edu.au}
}
\affil[2]{
     \small{University of Vienna, Austria.}
     \small{martin.seybold@univie.ac.at}
}
\affil[3]{
     \small{University of Copenhagen, Denmark.}
     \small{sawo@di.ku.dk}
}

\date{}

\begin{document}

\maketitle

\begin{abstract}
Map matching is a common preprocessing step for analysing vehicle trajectories. In the theory community, the most popular approach for map matching is to compute a path on the road network that is the most spatially similar to the trajectory, where spatial similarity is measured using the Fr\'echet distance. A shortcoming of existing map matching algorithms under the Fr\'echet distance is that every time a trajectory is matched, the entire road network needs to be reprocessed from scratch. An open problem is whether one can preprocess the road network into a data structure, so that map matching queries can be answered in sublinear time.

In this paper, we investigate map matching queries under the Fr\'echet distance. We provide a negative result for geometric planar graphs. We show that, unless SETH fails, there is no data structure that can be constructed in polynomial time that answers map matching queries in $O((pq)^{1-\delta})$ query time for any $\delta > 0$, where~$p$ and~$q$ are the complexities of the geometric planar graph and the query trajectory, respectively. We provide a positive result for realistic input graphs, which we regard as the main result of this paper. We show that for $c$-packed graphs, one can construct a data structure of $\tilde O(cp)$ size that can answer $(1+\varepsilon)$-approximate map matching queries in $\tilde O(c^4 q \log^4 p)$ time, where $\tilde O(\cdot)$ hides lower-order factors and dependence on~$\varepsilon$.
\end{abstract}

\section{Introduction}

Location-aware devices have enabled the tracking of vehicle trajectories. In urban environments, vehicle trajectories align with an underlying road network. However, imprecision in the Global Positioning System introduces errors into the trajectory data. Map matching aims to mitigate the effects of these errors by computing a path on the underlying road network that best represents the vehicle's trajectory. See Figure~\ref{fig:introduction1}.

\begin{figure}[ht]
	\centering
	\includegraphics{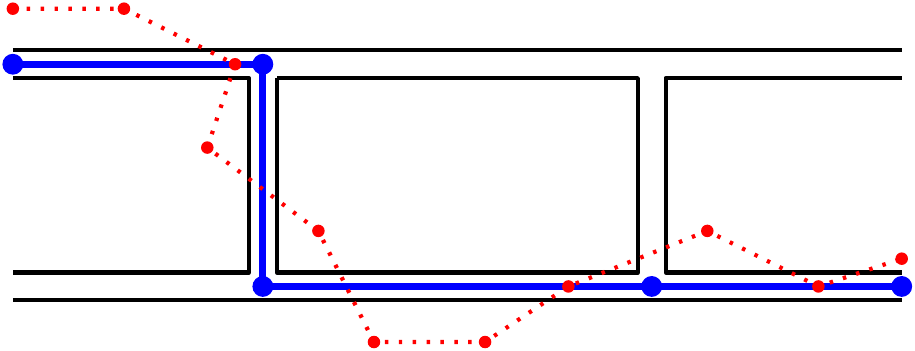}
	\caption{A road network (black), a noisy trajectory (red), and its matched path (blue).}
	\label{fig:introduction1}
\end{figure}

Map matching is a common preprocessing step for analysing vehicle trajectories. As such, numerous map matching algorithms have been proposed across multiple communities (e.g. in the Urban Planning, Geographic Information Systems, and Databases communities). Map matching was the focus of the 2012 ACM SIGSPATIAL Cup~\cite{DBLP:conf/gis/AliKRT12}. For an overview of the extensive literature on map matching, see the surveys~\cite{DBLP:conf/adc/ChaoXH020,DBLP:journals/urban/HashemiK14,DBLP:journals/itsm/KubickaCMN18,quddus2007current,DBLP:conf/gis/WeiWFZ13,DBLP:books/sp/zheng11}. In the theory community, by far the most popular approach is to embed the road network and the trajectory into the Euclidean plane, and to compute a path on the road network that is the most spatially similar to the trajectory \cite{DBLP:journals/jal/AltERW03,DBLP:conf/vldb/BrakatsoulasPSW05,DBLP:journals/tsas/ChambersFWW20,DBLP:conf/alenex/ChenDGNW11,DBLP:journals/comgeo/GudmundssonS15,DBLP:conf/gis/ShigezumiAMI15,DBLP:conf/ssdbm/WenkSP06}, where spatial similarity is measured using the Fr\'echet distance~\cite{DBLP:journals/ijcga/AltG95}. Formally, the map matching problem under the Fr\'echet distance is defined as follows. 

\begin{problem}[Map matching]
\label{problem:map_matching}
Given a graph~$P$ and a trajectory~$Q$ in the plane, compute a path~$\pi$ in~$P$ that minimises $\frechet(\pi,Q)$, where $\frechet(\cdot,\cdot)$ denotes the Fr\'echet distance.
\end{problem}

In a seminal paper by Alt, Efrat, Rote and Wenk~\cite{DBLP:journals/jal/AltERW03}, the authors study Problem~\ref{problem:map_matching} on geometric planar graphs. They provide an $O(pq \log p)$ time algorithm, where~$p$ is the complexity of the graph and~$q$ is the complexity of the trajectory. Their idea is to construct a free space surface, which is a generalisation of the free space diagram~\cite{DBLP:journals/ijcga/AltG95}, and then to perform a sweep line algorithm where a set of reachable points is maintained at the sweep line's current position.

Alt~et~al.~\cite{DBLP:journals/jal/AltERW03}'s algorithm forms the basis of several existing implementations~\cite{DBLP:conf/vldb/BrakatsoulasPSW05,DBLP:conf/alenex/ChenDGNW11,DBLP:conf/gis/ShigezumiAMI15,wei2013map,DBLP:conf/ssdbm/WenkSP06}. Brakatsoulas~et~al.~\cite{DBLP:conf/vldb/BrakatsoulasPSW05} implement Alt~et~al.~\cite{DBLP:journals/jal/AltERW03}'s algorithm and experimentally compare it to a linear-time heuristic and an algorithm minimising the weak Fr\'echet distance. In their experiments, forty-five vehicle trajectories, each with approximately one hundred edges, are mapped onto an underlying road network with approximately ten thousand edges. Their experiments conclude that out of the three algorithms, Alt~et~al.~\cite{DBLP:journals/jal/AltERW03}'s provides the best map matching results but is the slowest. Subsequent papers focus on improving the practical running time of the algorithm~\cite{DBLP:conf/gis/ShigezumiAMI15,DBLP:conf/ssdbm/WenkSP06}. We show that a significantly faster algorithm for geometric planar map matching is unlikely to exist, unless SETH fails.

Traditional analysis focuses on worst case instances, which are unlikely to occur in practice. By making realistic input assumptions, we can circumvent these worst-case instances, and provide bounds that better reflect running time on realistic input. In computational movement analysis, the most popular realistic input assumption is $c$-packedness. A set of edges is $c$-packed if the total length of edges inside any ball is at most~$c$ times the radius of the ball. Given two $c$-packed trajectories of complexity $n$, one can $(1+\varepsilon)$-approximate their Fr\'echet distance in $O(c\varepsilon^{-1/2} \log(1/\varepsilon)n + cn \log n)$ time~\cite{DBLP:journals/ijcga/BringmannK17,DBLP:journals/dcg/DriemelHW12}, circumventing the $\Omega(n^{2-\delta})$ lower bound for all $\delta > 0$ implied by SETH~\cite{DBLP:conf/focs/Bringmann14,DBLP:conf/soda/BuchinOS19}. 

Chen, Driemel, Guibas, Nguyen and Wenk~\cite{DBLP:conf/alenex/ChenDGNW11} study map matching on realistic input graphs and realistic input trajectories. They provide a $(1+\varepsilon)$-approximation algorithm that runs in $O((p+q) \log (p+q) + (\phi q + cp) \log pq \log (p+q) + (\phi \varepsilon^{-2} q + c \varepsilon^{-1} p) \log (pq))$ time, where the graph is $\phi$-low-density and has complexity~$p$, and the trajectory is $c$-packed and has complexity~$q$. A graph is $\phi$-low-density if, for any ball of radius $r$, the number of edges with length at least $r$ that intersect the ball is at most $\phi$. Note that a $c$-packed graph is $2c$-low-density~\cite{DBLP:journals/dcg/DriemelHW12}. Chen~et~al.~\cite{DBLP:conf/alenex/ChenDGNW11} implement their algorithm to map trajectories, each with at most one hundred edges, onto an underlying road network with approximately one million edges. Their experiments show significant running time improvements compared to previous work. On large road networks, their map matching algorithm runs in under a second, whereas previous algorithms~\cite{DBLP:conf/vldb/BrakatsoulasPSW05,DBLP:conf/ssdbm/WenkSP06} require several hours.

For both general graphs and realistic input graphs, a shortcoming of the existing map matching algorithms is that, every time a trajectory is matched, the entire road network needs to be reprocessed from scratch. This is reflected in the linear dependence on~$p$ in the running times of Alt~et~al.~\cite{DBLP:journals/jal/AltERW03} and Chen~et~al.~\cite{DBLP:conf/alenex/ChenDGNW11}, where~$p$ is the complexity of the input graph. If we were to build a data structure for efficient map matching queries, then we would remove the need to reprocess the graph every time a new trajectory is mapped. Formally, the query version of the map matching problem is given below.

\begin{restatable}[Map matching queries]{problem}{mapmatchingqueriesproblem}
\label{problem:mapmatchingqueriesproblem}
Given a graph~$P$ in the plane, construct a data structure so that given a query trajectory~$Q$ in the plane, the data structure returns $\min_{\pi} \frechet(\pi,Q)$, where $\pi$ ranges over all paths in~$P$ and $\frechet(\cdot,\cdot)$ denotes the Fr\'echet distance.
\end{restatable}

To the best of our knowledge, our paper is the first to study map matching queries under the Fr\'echet distance. Obtaining a data structure for Problem~\ref{problem:mapmatchingqueriesproblem} with query time that is sublinear in the complexity of the graph is stated as an open problem in~\cite{DBLP:conf/alenex/ChenDGNW11} and in~\cite{DBLP:journals/comgeo/GudmundssonS15}. 

\subsection{Contributions}

In this paper, we investigate map matching queries under the Fr\'echet distance. An open problem proposed independently by~\cite{DBLP:conf/alenex/ChenDGNW11} and~\cite{DBLP:journals/comgeo/GudmundssonS15} asks whether it is possible to preprocess a graph into a data structure so that map matching queries can be answered in sublinear time.

We provide a negative result in the case of geometric planar graphs. We show that, unless SETH fails, there is no data structure that can be constructed in polynomial preprocessing time, that answers map matching queries in $O((pq)^{1-\delta})$ query time for any $\delta > 0$, where~$p$ and~$q$ are the complexities of the graph and the query trajectory, respectively. Our negative result shows that preprocessing does not help for geometric planar map matching.

We provide the first positive result in the case of realistic input graphs. Our data structure has near-linear size in terms of~$p$, and its query time is polylogarithmic in terms of~$p$. We consider the following theorem to be the main result of our paper.

\begin{restatable}{theorem}{main}
\label{theorem:main}
Given a $c$-packed graph~$P$ of complexity~$p$, one can construct a data structure of $O(p \log^2 p + c \varepsilon^{-4} \log(1/\varepsilon) p \log p)$ size, so that given a query trajectory~$Q$ of complexity~$q$, the data structure returns in $O(q \log q \cdot (\log^4 p + c^4 \varepsilon^{-8} \log^2 p))$ query time a $(1+\varepsilon)$-approximation of $\min_{\pi} \frechet(\pi,Q)$ where $\pi$ ranges over all paths in~$P$ and $\frechet(\cdot,\cdot)$ denotes the Fr\'echet distance. The preprocessing time is~$O(c^2 \varepsilon^{-4} \log^2(1/\varepsilon) p^2 \log^2 p)$.
\end{restatable}

The most closely related results are~\cite{DBLP:conf/alenex/ChenDGNW11} and~\cite{DBLP:journals/comgeo/GudmundssonS15}. We briefly compare the realistic input assumptions of these related works to our result. In Chen~et~al.~\cite{DBLP:conf/alenex/ChenDGNW11}, the graph is $\phi$-low-density and the trajectory is $c$-packed. In Gudmundsson and Smid~\cite{DBLP:journals/comgeo/GudmundssonS15}, the graph is a $c$-packed tree with long edges, and the trajectory has long edges. In our result, the graph is $c$-packed, but surprisingly, we require no input assumptions on the query trajectory.

\subsection{Related work}

The Fr\'echet distance is a popular similarity measure for trajectories. To compute the Fr\'echet distance between a pair of trajectories of complexity $n$, Alt and Godau~\cite{DBLP:journals/ijcga/AltG95} provide an $O(n^2 \log n)$ time algorithm, which Buchin~et~al.~\cite{DBLP:journals/dcg/BuchinBMM17} improve to an $O(n^2 \sqrt{\log n} (\log \log n)^{3/2})$ algorithm. Conditioned on the Strong Exponential Time Hypothesis (SETH), for all $\delta > 0$, Bringmann~\cite{DBLP:conf/focs/Bringmann14} shows an $\Omega(n^{2-\delta})$ lower bound for computing the Fr\'echet distance in two or more dimensions. Buchin~et~al.~\cite{DBLP:conf/soda/BuchinOS19} generalise the lower bound to one or more dimensions. 

Variants of Problem~\ref{problem:map_matching} have been considered. Seybold~\cite{DBLP:conf/sdm/Seybold17} and Chambers~et~al.~\cite{DBLP:journals/tsas/ChambersFWW20} consider finding the shortest map matching paths in geometric graphs. Chen~et~al.~\cite{DBLP:conf/atmos/ChenSW21} study map matching under the weak Fr\'echet distance, whereas Wylie and Zhu~\cite{DBLP:journals/corr/WylieZ14} and Fu~et~al.~\cite{DBLP:conf/cccg/FuSW19} consider map matching under the discrete Fr\'echet distance. Wei~et~al.~\cite{wei2013map} and Chen~et~al.~\cite{DBLP:conf/atmos/ChenSW21} combine the Fr\'echet distance approach with a Hidden Markov Model approach to obtain a hybrid algorithm.

A problem closely related to Problem~\ref{problem:mapmatchingqueriesproblem} is to preprocess a trajectory for Fr\'echet distance queries. Driemel and Har-Peled~\cite{DBLP:journals/siamcomp/DriemelH13} preprocess a trajectory $Z$ of complexity $n$ in $O(n \log^3 n)$ time and $O(n \log n)$ space, so that given a query trajectory~$Q$ with complexity $k$, and a query subtrajectory~$Z[u,v]$ where~$u$ and~$v$ are points on $Z$, one can return in $O(k^2 \log n \log(k \log n))$ time a constant factor approximation of the Fr\'echet distance between $Z[u,v]$ and~$Q$. In this paper we show that, even with polynomial time preprocessing time on the trajectory~$Z$, one cannot hope to answer Fr\'echet distance queries in truly subquadratic time, unless SETH fails. As such, special cases have been considered. For $k=2$, one can answer $(1+\varepsilon)$-approximate~\cite{DBLP:journals/siamcomp/DriemelH13} or exact~\cite{DBLP:conf/esa/BuchinHOSSS22,DBLP:conf/cccg/BergMO17,DBLP:journals/algorithmica/GudmundssonRSW21} queries in polylogarithmic query time, by constructing a data structure of subquadratic size. Discrete Fr\'echet distance queries for small values of~$k$ have also been studied~\cite{DBLP:journals/corr/driemelpsarros,DBLP:conf/soda/FiltserF21,DBLP:journals/ipl/Filtser18}.

Gudmundsson and Smid~\cite{DBLP:journals/comgeo/GudmundssonS15} preprocess a $c$-packed tree for Fr\'echet distance queries. We regard this to be one of the most relevant results to our work. Given a $c$-packed tree~$T$, and a positive real number~$\Delta$, the authors show how to construct a data structure of size $O(cn)$ in $O(n \log^2 n + cn\log n)$ preprocessing time, so that given a polygonal curve~$Q$ with $k$ vertices, one can decide in $O(c^4 k \log^2 n)$ time whether there exists a path $\pi \in T$ so that $\frechet(\pi,Q) \leq 3.001 \cdot \Delta$, or that $\frechet(\pi,Q) > \Delta$ for all paths $\pi \in T$, where $\frechet(\cdot,\cdot)$ denotes the Fr\'echet distance. The authors assume that the edges of $T$ and~$Q$ have length~$\Omega(\Delta)$. Since $\Delta$ is fixed at preprocessing time, it is unclear whether it is possible to minimise the Fr\'echet distance to solve Problem~\ref{problem:mapmatchingqueriesproblem}.

Related structures that have received considerable attention include range searching and approximate nearest neighbour searching under the Fr\'echet distance~\cite{DBLP:conf/soda/AfshaniD18,DBLP:conf/gis/BaldusB17,DBLP:conf/soda/BringmannDNP22,DBLP:conf/gis/BuchinDDM17,DBLP:conf/gis/BergGM17,DBLP:conf/wads/DriemelP21,DBLP:conf/gis/DutschV17,DBLP:conf/icalp/FiltserFK20,DBLP:journals/tsas/GudmundssonHPS21,DBLP:conf/compgeom/Indyk02}. 

\section{Preliminaries}
\label{sec:preliminaries}

Let $P = (V,E)$ be an undirected graph embedded in the Euclidean plane~$\mathbb R^2$. An edge $uv \in E$ is a segment between $u,v \in V$, with length equal to the Euclidean distance, i.e. $|uv| = d(u,v)$. Let $p = |V| + |E|$ be the complexity of the graph~$P$. We assume that~$P$ is connected, otherwise, our map matching queries can be handled for each connected component independently. A path $\pi \in P$ is defined to be a sequence of vertices $u_1, \ldots, u_k \in V$ so that $u_i u_{i+1} \in E$ for all $1 \leq i < k$. In particular, for the purposes of this paper we consider only paths $\pi$ in~$P$ that start and end at vertices of~$P$. Given a pair of vertices $u, v \in P$, the graph metric~$d_P$ is defined so that $d_P(u,v)$ equals the total length of the shortest path between $u$ and $v$ in the graph~$P$. The graph~$P$ is $c$-packed if, for every ball $B_r$ of radius $r$ in the Euclidean plane, the total length of edges in $E$ inside $B_r$ is upper bounded by~$cr$. Formally, $\sum_{e \in E} |e \cap B_r| \leq c r$.

A trajectory is a sequence of vertices in the Euclidean plane. Given vertices~$a_1, \ldots, a_q$, the polygonal curve~$Q$ is a piecewise linear function $Q:[1,q] \to \mathbb R^2$ satisfying $Q(i) = a_i$ for all $1 \leq i \leq q$, and $Q(i + \mu) = (1-\mu) a_i + \mu a_{i+1}$ for all integers $1 \leq i \leq q-1$ and reals $0 \leq \mu \leq 1$. Let $\Gamma(q)$ be the space of all continuous non-decreasing surjective functions for $[0,1] \to [1,q]$. For a pair of polygonal curves~$Q_1$ and~$Q_2$ of complexities $n_1$ and $n_2$, we define the Fr\'echet distance between $Q_1$ and $Q_2$ to be $\frechet(Q_1,Q_2) = \inf_{(\alpha_1, \alpha_2) \in \Gamma(n_1) \times \Gamma(n_2)} \ \max_{\mu \in [0,1]} d(Q_1(\alpha_1(\mu)), Q_2(\alpha_2(\mu)))$, where $d(\cdot,\cdot)$ denotes the Euclidean distance.

Let $0 < \varepsilon < 1$ be a constant that is fixed at preprocessing time.

\section{Technical Overview}
\label{sec:overview}

In Section~\ref{sec:overview;subsec:data_structure}, we give an overview of our data structure for map matching queries on $c$-packed graphs. In Section~\ref{sec:overview;subsec:lower_bound}, we give an overview of our lower bound for map matching queries on geometric planar graphs. Full proofs are provided in Sections~\ref{sec:straightest_path_queries}-\ref{sec:lower_bounds}.

\subsection{Data structure for \texorpdfstring{$c$}{c}-packed graphs}
\label{sec:overview;subsec:data_structure}

Our data structure for $c$-packed graphs is built in three stages. In Stage~1, we construct a data structure for straightest path queries, which we will define in due course. In Stage~2, we construct a data structure for map matching queries, in the special case that the trajectory is a segment. In Stage~3, we construct a data structure for map matching queries in general. Each stage builds upon and generalises the previous stage. We provide an overview of Stages~1,~2 and~3 in Sections~\ref{sec:overview;subsubsec:straightest_path_queries},~\ref{sec:overview;subsubsec:map_matching_segment_queries} and~\ref{sec:overview;subsubsec:map_matching_queries} respectively.

\subsubsection{Stage 1: Straightest path queries}
\label{sec:overview;subsubsec:straightest_path_queries}

For every pair of vertices in the graph, we are interested in precomputing a path between them that is as straight as possible. We define straightness using the Fr\'echet distance. Formally, given a pair of vertices~$u$ and~$v$, we define a straightest path between~$u$ and~$v$ to be a path $\pi \in P$ between~$u$ and~$v$ that minimises the Fr\'echet distance $\frechet(\pi,uv)$. This leads us to the following definition for straightest path queries.

\begin{subproblem}[Straightest path queries]
\label{problem:straightest_path_queries}
Given a graph~$P$ in the plane, construct a data structure so that given any pair of vertices $u,v \in P$, the data structure returns $\min_{\pi} \frechet(\pi,uv)$ where~$\pi$ ranges over all paths in~$P$ between~$u$ and~$v$. See Figure~\ref{fig:technical_overview_1}.
\end{subproblem}

\begin{figure}[ht]
	\centering
	\includegraphics{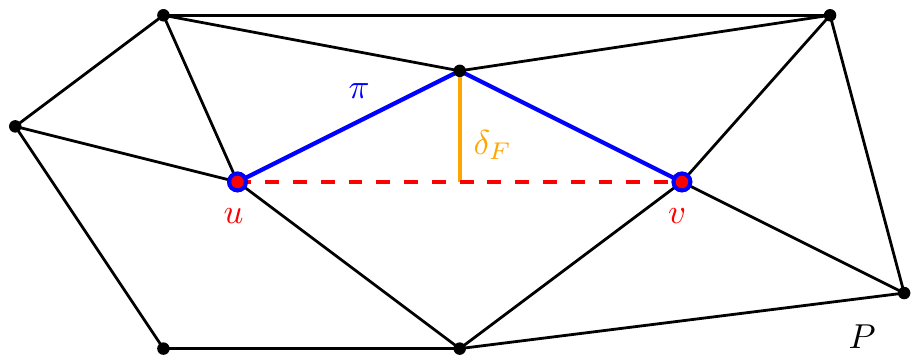}
	\caption{Given a pair of query vertices $u,v$ (red), the data structure in Subproblem~\ref{problem:straightest_path_queries} returns $\min_{\pi} \frechet(\pi, uv)$ (orange) where $\pi$ ranges over all paths between~$u$ and~$v$ (blue) in the graph~$P$~(black).} 
	\label{fig:technical_overview_1}
\end{figure}

Note that Subproblem~\ref{problem:straightest_path_queries} can be viewed as a special case of map matching queries, where the query trajectory must be a segment between two graph vertices, and the path must connect the endpoints of the query segment. A na\"ive way to answer straightest path queries is, for every pair of vertices, to precompute the Fr\'echet distance for its straightest path. Unfortunately, storing the precomputed Fr\'echet distance for all pairs of vertices requires $\Omega(p^2)$ space. 

Instead, we use a semi-separated pair decomposition~\cite{DBLP:journals/dcg/AbamBFG09} to reduce the number of pairs we need to consider. We define the transit vertices of a semi-separated pair to be a set of vertices so that any path between the two components of the semi-separated pair must pass through at least one of the transit vertices. We define the set of transit pairs of a semi-separated pair to be pairs of vertices where one vertex is in the semi-separated pair, and one vertex is a transit vertex. For $c$-packed graphs, we show that there are at most $O(c p \log p)$ transit pairs. By storing the minimum Fr\'echet distance for each transit pair, we reduce the storage requirement of our data structure to $O(c p \log p)$. 

Finally, we answer straightest path queries by dividing the path $u \to v$ into two paths. Specifically, we divide $u \to v$ into $u \to w \to v$, where~$w$ is a transit vertex of the semi-separated pair separating~$u$ and~$v$. Having precomputed the minimum Fr\'echet distance for transit pairs~$(u,w)$ and~$(w,v)$, we use these Fr\'echet distances to obtain a constant factor approximation for the Fr\'echet distance of the straightest path between~$u$ and~$v$. 

Putting this all together, we obtain Theorem~\ref{theorem:straightestpathdatastructure}. For a full proof see Section~\ref{sec:straightest_path_queries}.

\begin{restatable}{theorem}{straightestpathdatastructure}
\label{theorem:straightestpathdatastructure}
Given a $c$-packed graph~$P$ of complexity~$p$, one can construct a data structure of $O(c p \log p)$ size, so that given a pair of query vertices $u,v \in P$, the data structure returns in $O(\log p)$ query time a $3$-approximation of $\min_{\pi} \frechet(\pi,uv)$, where $\pi$ ranges over all paths in~$P$ between~$u$ and~$v$. The preprocessing time is $O(c p^2 \log^2 p)$.
\end{restatable}

\subsubsection{Stage 2: Map matching segment queries}
\label{sec:overview;subsubsec:map_matching_segment_queries}

Our next step is to answer map matching queries where the query trajectory is an arbitrary segment. 

\begin{subproblem}[Map matching segment queries]
\label{problem:map_matching_segment_queries}
Given a graph~$P$ in the plane, construct a data structure so that given a query segment~$Q$ in the plane, the data structure returns $min_\pi \frechet(\pi,Q)$ where $\pi$ ranges over all paths in~$P$ that start and end at a vertex of~$P$. See Figure~\ref{fig:technical_overview_2}.
\end{subproblem}

\begin{figure}[ht]
	\centering
	\includegraphics{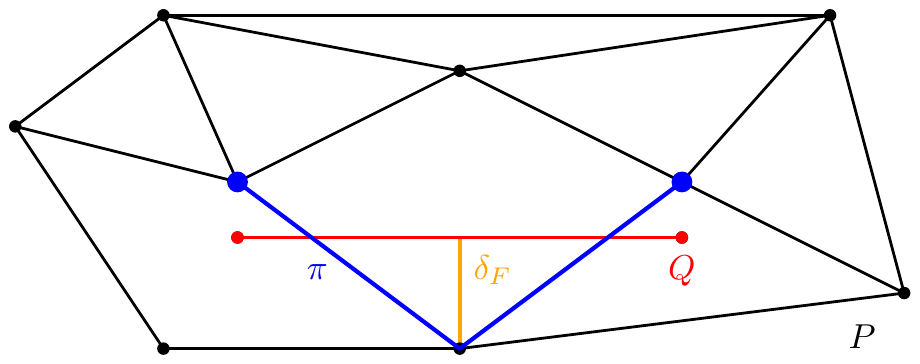}
	\caption{Given a query segment~$Q$ (red), the data structure in Subproblem~\ref{problem:map_matching_segment_queries} returns $\min_{\pi} \frechet(\pi, Q)$ (orange) where $\pi$ ranges over all paths (blue) in the graph~$P$ (black).} 
	\label{fig:technical_overview_2}
\end{figure}

Subproblem~\ref{problem:map_matching_segment_queries} can be viewed as a generalisation of Subproblem~\ref{problem:straightest_path_queries}, where the endpoints of the segment~$Q$ are not necessarily graph vertices, and the starting and ending points of the path are not given and must instead be computed. To answer map matching segment queries, we combine two data structures. The first data structure is an extension of the data structure in Theorem~\ref{theorem:straightestpathdatastructure}, which we modify to handle query segments that do not necessarily have their endpoints at graph vertices. The second data structure is to build a simplification of the $c$-packed graph so that one can efficiently query the starting and ending points of the path.

To build our first data structure, we use the result of Driemel and Har-Peled~\cite{DBLP:journals/siamcomp/DriemelH13}, which states that one can preprocess a trajectory in near-linear time and space, so that given a query segment, one can $(1+\varepsilon)$-approximate the Fr\'echet distance from the query segment to any subcurve of the trajectory in constant time. Let $\varepsilon > 0$ and $\chi = \varepsilon^{-2} \log (1/\varepsilon)$. By combining Theorem~\ref{theorem:straightestpathdatastructure} with their result, we obtain a data structure of $O(c\chi^2 p \log p)$ size, so that given a pair of query vertices~$u,v \in P$ and a query segment $ab$ in the plane, the data structure returns in $O(\log p + c \varepsilon^{-1})$ time a $(1+\varepsilon)$-approximation of $\min_{\pi} \frechet(\pi, ab)$, where $\pi$ ranges over all paths in~$P$ between~$u$ and~$v$. The preprocessing time is $O(c \chi^2 p^2 \log^2 p)$. 

To build our second data structure, we use graph clustering to simplify the $c$-packed graph. We first consider the decision version of the problem. Given a Fr\'echet distance $r$, we guarantee that all edges in our simplified graph have length at least~$\varepsilon r$. By $c$-packedness, the number of simplified graph vertices inside a disk of radius $r$ is at most a constant. Given a query segment $ab$, this reduces the number of candidate starting and ending points of the matched path to a constant. We use an orthogonal range searching data structure to efficiently query for the candidate starting and ending points of the path. Finally, we apply parametric search to minimise the Fr\'echet distance $r$.

By combining our two data structures, we obtain Theorem~\ref{theorem:mapmatchingsegmentqueries}. For a full proof see Section~\ref{sec:map_matching_segment_queries}.

\begin{restatable}{theorem}{mapmatchingsegmentqueries}
\label{theorem:mapmatchingsegmentqueries}
Given a $c$-packed graph~$P$ of complexity~$p$, one can construct a data structure of $O(c \varepsilon^{-4}\log^2(1/\varepsilon) \cdot p \log p )$ size, so that given a query segment $ab$ in the plane, the data structure returns in $O(c^4 \varepsilon^{-4} \cdot \log^2 p)$ time a $(1+\varepsilon)$-approximation of $\min_{\pi} \frechet(\pi,ab)$, where $\pi$ ranges over all paths in~$P$ that start and end at a vertex of~$P$. The preprocessing time is $O(c\varepsilon^{-4}\log^2(1/\varepsilon) \cdot p^2 \log^2 p)$.
\end{restatable}

\subsubsection{Stage 3: Map matching queries}
\label{sec:overview;subsubsec:map_matching_queries}

Finally, we consider general map matching queries, which we restate for convenience. See Figure~\ref{fig:technical_overview_3}.

\mapmatchingqueriesproblem*

\begin{figure}[ht]
	\centering
	\includegraphics{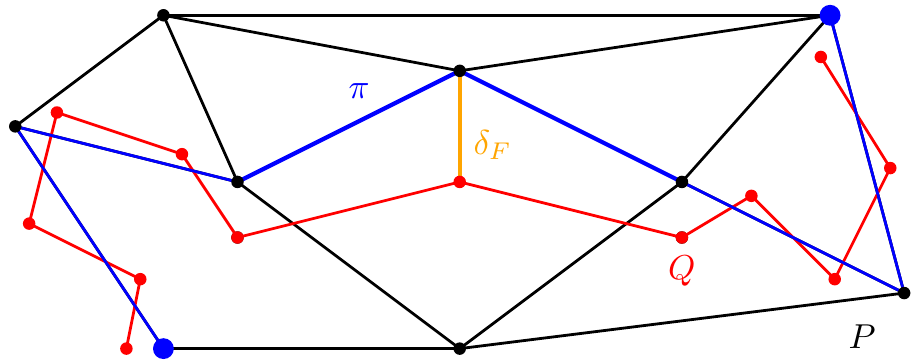}
	\caption{Given a query trajectory~$Q$ (red), the data structure in Problem~\ref{problem:mapmatchingqueriesproblem} returns $\min_{\pi} \frechet(\pi, Q)$ (orange) where $\pi$ ranges over all paths (blue) in the graph~$P$ (black).} 
	\label{fig:technical_overview_3}
\end{figure}

Let the vertices of~$Q$ be $a_1, \ldots, a_q$, and let the Fr\'echet distance for the decision version be~$r$. We compute a set of points $T_i$ that $a_i$ can match to. Each point $b_{i,j} \in T_i$ is either a vertex of $P$, or a point along an edge of $P$. The size of $T_i$ is at most a constant, depending on $c$ and $\varepsilon$. For a point~$b_{i,j}$ that is a vertex of~$P$, we construct it in the same way as in Stage~2. For a point~$b_{i,j}$ that is on an edge of~$P$, we construct it by sampling along the edges of~$P$ that have length at least $\varepsilon r/2$ and are within a distance of~$r$ to~$a_i$. To efficiently query the edges of~$P$ with these two properties, we build a three-dimensional low-density environment and use the range searching data structure of Schwarzkopf and Vleugels~\cite{DBLP:journals/ipl/SchwarzkopfV96}.

The final step is to build a directed graph on the set of points $\cup_{i=1}^q T_i$. For each~$b_{i,j} \in T_i$ and~$b_{i+1,k} \in T_{i+1}$, we use the map matching segment query from the previous section to compute a $(1+\varepsilon)$-approximation of the minimum Fr\'echet distance $\min_{\pi} (\pi, a_ia_{i+1})$ where $\pi$ ranges over all paths between~$b_{i,j}$ and~$b_{i+1,k}$. We set the capacity of the directed edge from~$b_{i,j}$ to~$b_{i+1,k}$ to be this minimum Fr\'echet distance. We decide whether there is a directed path from a point in $T_1$ to a point in $T_q$ so that the capacities of all edges on the directed path are at most $(1+\varepsilon) r$. Finally, we use parametric search to minimise~$r$. 

Putting this all together, we obtain Theorem~\ref{theorem:main}, which we restate for convenience. For a full proof see Section~\ref{sec:map_matching_queries}.

\main*

\subsection{Lower bound for geometric planar graphs}
\label{sec:overview;subsec:lower_bound}

In the final section, we investigate lower bounds for map matching queries on graphs that are not $c$-packed. Our lower bounds attempt to explain why answering map matching queries is such a difficult problem in general. In particular, we show that unless SETH fails, there is no data structure that can be constructed in polynomial preprocessing time, that answers map matching queries on geometric planar graphs in truly subquadratic time. Note that the upper bound of Alt~et~al.~\cite{DBLP:journals/jal/AltERW03} matches this lower bound up to lower-order factors, which implies that preprocessing does not help for geometric planar map matching, unless SETH fails. 

To build towards our lower bound for map matching queries, we consider a warm-up problem, which is to preprocesses a trajectory, so that given a query trajectory, the data structure can efficiently answer the Fr\'echet distance between the query trajectory and the preprocessed trajectory. Buchin~et~al.~\cite{DBLP:conf/esa/BuchinHOSSS22} claim that this is an extremely difficult problem, which is why the special case of query segments is considered in their paper. We provide evidence towards Buchin~et~al.'s~\cite{DBLP:conf/esa/BuchinHOSSS22} claim. We show that preprocessing does not help with Fr\'echet distance queries on trajectories unless SETH fails. In particular, there is no data structure with polynomial preprocessing time that can answer Fr\'echet distance queries significantly faster than computing the Fr\'echet distance without preprocessing. To show our lower bound, we modify Bringmann's~\cite{DBLP:conf/focs/Bringmann14} construction to answer the offline version of the data structure problem in a similar fashion to Rubinstein~\cite{DBLP:conf/stoc/Rubinstein18}, Driemel and Psarros~\cite{DBLP:journals/corr/driemelpsarros} and Bringmann~et~al.~\cite{DBLP:conf/soda/BringmannDNP22}. 

Next, we prove a lower bound for Problem~\ref{problem:map_matching}. We show that unless SETH fails, there is no truly subquadratic time for map matching on geometric planar graphs. This shows that the algorithm by Alt~et~al.~\cite{DBLP:journals/jal/AltERW03} for geometric planar map matching is optimal up to lower-order factors, unless SETH fails. Finally, we combine the ideas from our warm-up problem and our lower bound for Problem~\ref{problem:map_matching} to rule out truly subquadratic query times for map matching queries on geometric planar graphs, unless SETH fails. 

Putting this all together, we obtain Theorem~\ref{theorem:lower_bound_map_matching_query}. For a full proof see Section~\ref{sec:lower_bounds}.

\begin{restatable}{theorem}{lowerboundmapmatchingquery}
\label{theorem:lower_bound_map_matching_query}
Given a geometric planar graph of complexity~$p$, there is no data structure that can be constructed in $\poly(p)$ time, that when given a query trajectory of complexity~$q$, can answer $2.999$-approximate map matching queries in $O((pq)^{1-\delta})$ query time for any $\delta > 0$, unless SETH fails. This holds for any polynomial restrictions of~$p$ and~$q$.
\end{restatable}

This completes the overview of the main results of our paper.

\section{Stage 1: Straightest path queries}
\label{sec:straightest_path_queries}

The first stage of our data structure for $c$-packed graphs is to construct a straightest path query data structure. Recall that the straightest path between~$u$ and~$v$ is a path $\pi \in P$ from~$u$ to~$v$ that minimises the Fr\'echet distance $\frechet(\pi,uv)$. A data structure for straightest path queries is defined as follows. Given a pair of query vertices~$u$ and~$v$, the data structure returns the minimum Fr\'echet distance $\frechet(\pi,uv)$ where $\pi$ ranges over all paths between~$u$ and~$v$. As stated in the technical overview, we avoid storing a quadratic number of Fr\'echet distances by using a semi-separated pair decomposition (SSPD) to reduce the number of pairs of vertices we need to consider.

\begin{definition}[SSPD]
Let~$V$ be a set of vertices. A semi-separated pair decomposition of~$V$ with separation constant $s \in \mathbb R^+$ is a collection $\{(A_i, B_i)\}_{i=1}^k$ of pairs of non-empty subsets of~$V$ so that $$\min(\diameter(A_i),\diameter(B_i)) \leq s \, \cdot d(A_i,B_i),$$ and for any two distinct points~$u$ and~$v$ of~$V$, there is exactly one pair $(A_i, B_i)$ in the collection, such that~$(i)$ $u \in A_i$ and $v \in B_i$, or~$(ii)$ $v \in A_i$ and $u \in B_i$. 
\end{definition}

Note that for sets $A,B$, we define $d(A,B) = \min_{(a,b) \in A \times B} d(a,b)$, where $d(a,b)$ denotes the Euclidean distance. The total weight of $\{(A_i, B_i)\}_{i=1}^k$ is defined as $\sum_{i=1}^k (|A_i| + |B_i|)$. Abam~et~al.~\cite{DBLP:journals/dcg/AbamBFG09} show how to construct an SSPD of~$V$ with separation constant $s$ in $O(ns^{-2} + n \log n)$ time, that has $O(ns^{-2})$ pairs, and total weight $O(ns^{-2} \log n)$, where $n$ is number of vertices in~$V$.

Although not explicitly stated in~\cite{DBLP:journals/dcg/AbamBFG09}, given any two distinct points $u$ and $v$ of $V$, one can query the SSPD in $O(\log n)$ time to retrieve the pair $(A_i,B_i)$ satisfying either $(i)$ $u \in A_i$ and $v \in B_i$ or $(ii)$ $v \in A_i$ and $u \in B_i$. We provide a sketch of the query procedure. The SSPD in~\cite{DBLP:journals/dcg/AbamBFG09} is constructed using a Balanced Aspect Ratio (BAR) tree~\cite{DBLP:journals/jal/DuncanGK01}, where each node in the balanced tree has an associated weight class. Each leaf of the BAR tree is associated with a point in~$V$, and has a weight class of~$O(\log n)$. Given any two distinct points $u$ and $v$, we simultaneously traverse the BAR tree, from the root to the leaf nodes associated with $u$ and $v$. The invariant maintained by the simultaneous traversal is that the weight class along the two traversals remains the same. The semi-separated pair $(A_i, B_i)$ that we return is the pair of nodes in the BAR tree with minimum weight class that satisfies the semi-separated property $\min(\diameter(A_i), \diameter(B_i) \leq s \cdot d(A_i, B_i)$. Putting this together, we obtain the following observation.

\begin{observation}
    \label{obs:query_sspd}
    Given a pair of distinct points $u$ and $v$ of $V$, one can query the SSPD of~\cite{DBLP:journals/dcg/AbamBFG09} in $O(\log n)$ time to obtain a semi-separated pair $(A_i, B_i)$ satisfying either (i) $u \in A_i$ and $v \in B_i$ or (ii) $v \in A_i$ and $u \in B_i$.
\end{observation}

We construct an SSPD of the vertices of~$P$ with separation constant $1/2$. For each semi-separated pair $(A_i, B_i)$ in our SSPD, we select a set of $O(c)$ vertices of~$P$ to be \emph{transit vertices}. The transit vertices have the property that any path from $A_i$ to $B_i$ must pass through a transit vertex. In Lemma~\ref{lemma:transit_vertices}, we show how to compute transit vertices.

\begin{lemma}
\label{lemma:transit_vertices}
Let $P = (V,E)$ be a $c$-packed graph and let $\{(A_i, B_i)\}_{i=1}^k$ be an SSPD of~$V$ with separation constant $1/2$. For each pair $(A_i, B_i)$ of the SSPD, one can compute a set of vertices~$C_i \subset V$ in $O(cp)$ time satisfying~$(i)$ $|C_i| \leq 2c$, and~$(ii)$ any path starting at a vertex in $A_i$ and ending at a vertex in $B_i$ must pass through a vertex in~$C_i$. 
\end{lemma}

\begin{proof}
First, we construct the set~$C_i$. Then we prove that~$C_i$ satisfies the required properties. Finally, we analyse the running time of our algorithm.

Without loss of generality, suppose $\diameter(A_i) \leq \diameter(B_i)$. Let $a_0$ be a vertex in $A_i$. Let~$D_1$ be a disk with centre at~$a_0$ with radius $\diameter(A_i)$, and let~$D_2$ be a disk with centre at~$a_0$ with radius $2\cdot \diameter(A_i)$. See Figure~\ref{fig:stage1_1}. All vertices of~$A_i$ are in~$D_1$. All vertices of~$B_i$ are outside~$D_2$, since $d(a_0, B_i) \geq d(A_i,B_i) \geq 2 \cdot \diameter(A_i) = \radius(D_2)$,
where the second inequality comes from the separation constant of the SSPD being $1/2$. 

\begin{figure}[ht]
	\centering
	\includegraphics{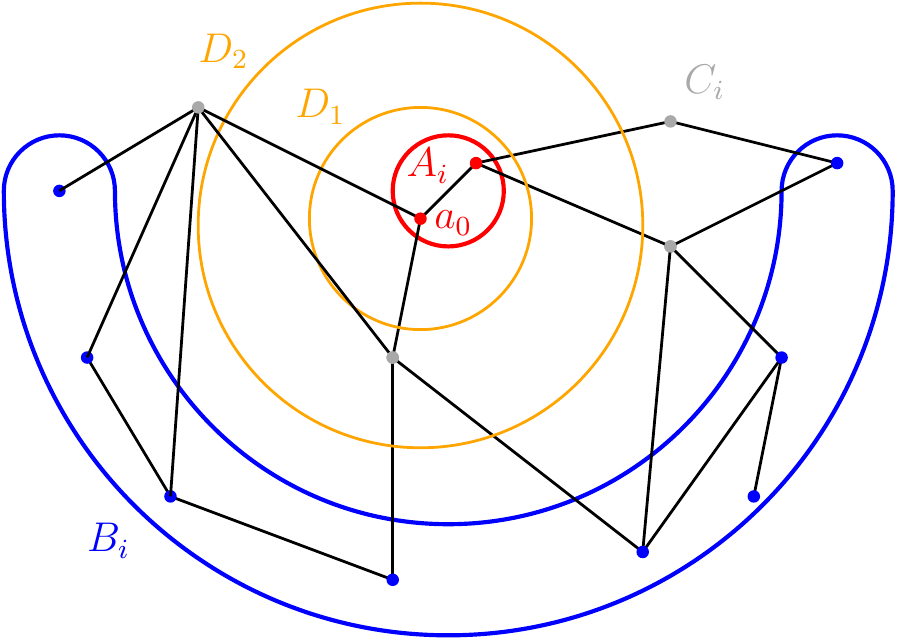}
	\caption{A semi-separated pair $A_i$ (red) and $B_i$ (blue). The circles $D_1$ and $D_2$ (orange) are centred at a vertex $a_0 \in A_i$, and have radius $\diameter(A_i)$ and $2 \cdot \diameter(A_i)$ respectively. The value of the max-flow/min-cut in the figure is $\ell = 4$, so $|C_i| = 4$ (grey).}
	\label{fig:stage1_1}
\end{figure}

Next, we set up a max-flow instance. Set the capacity of each edge of $P = (V,E)$ to 1. Set the vertices in $A_i$ to be sources, and set the vertices in $B_i$ to be sinks. The max-flow of the instance is equal to its min-cut. Let the minimum cut be a set of edges $e_1, e_2, \ldots, e_\ell$. Choose one endpoint for each edge $e_1, e_2, \ldots, e_\ell$ to form the set~$C_i$. This completes the construction of~$C_i$.

We show that our construction of~$C_i$ satisfies the properties~$(i)$ $|C_i| \leq 2c$, and~$(ii)$ any path starting at a vertex in $A_i$ and ending at a vertex in $B_i$ must pass through a vertex in~$C_i$. Property~$(ii)$ follows from $e_1, e_2, \ldots, e_\ell$ being a cut. This is because removing all the edges in the cut would disconnect the sources from the sinks, so all paths from $A_i$ to $B_i$ must pass through one of $e_1, e_2, \ldots, e_\ell$ and one of the vertices in~$C_i$. Property~$(i)$ follows from $c$-packedness. In the max-flow instance, the capacity of the max-flow is $\ell$. Since all edges have capacity~$1$, there are $\ell$ edge-disjoint paths from $A_i$ to~$B_i$. Each edge-disjoint path has one endpoint in $D_1$, and one endpoint outside $D_2$. So each path intersects both the inner and outer boundaries of the annulus $D_2 \setminus D_1$. The width of the annulus $D_2 \setminus D_1$ is equal to $\diameter(A_i)$. Therefore, there are $\ell$ edge disjoint paths in $D_2 \setminus D_1$ that each have length at least $\diameter(A_i)$. Since the graph is $c$-packed, the total length of edges in the ball $D_2$ is at most~$c$ times the radius of $D_2$, which is $2c \cdot \diameter(A_i)$. Therefore, $\ell \cdot \diameter(A_i) \leq 2c \cdot \diameter(A_i)$. Hence, $|C_i| = \ell \leq 2c$ as required.  

Finally, we analyse the running time of our algorithm, which is dominated by computing the max-flow. The running time of the Ford-Fulkerson algorithm is equal to the number of edges in~$P$ times the max-flow. Since $\ell \leq 2c$, the max-flow is $\leq 2c$. Moreover, there are $O(p)$ edges in~$P$. Therefore, the overall running time of the algorithm is $O(cp)$.
\end{proof}

The set of transit vertices for a semi-separated pair $(A_i,B_i)$ is defined to be the set~$C_i$ constructed above. Next, we define transit pairs. Given a semi-separated pair $(A_i,B_i)$, a transit pair for the semi-separated pair $(A_i,B_i)$ is a pair of vertices $(u,w)$ so that $u \in A_i \cup B_i$ and $w \in C_i$, where~$C_i$ is the set of transit vertices for $(A_i,B_i)$ defined in Lemma~\ref{lemma:transit_vertices}. Now, we bound the total number of transit vertices and pairs.

\begin{lemma}
\label{lemma:number_of_transit_vertices_and_transit_pairs}
There are $O(c p)$ transit vertices and $O(cp \log p)$ transit pairs in~$P$, over all semi-separated pairs in the SSPD. 
\end{lemma}

\begin{proof}
There are $O(p)$ semi-separated pairs in the SSPD in~\cite{DBLP:journals/dcg/AbamBFG09}. By Lemma~\ref{lemma:transit_vertices}, there are $O(c)$ transit vertices per semi-separated pair. Therefore, there are $O(cp)$ transit vertices in total. For a semi-separated pair $(A_i,B_i)$, let $(u,w)$ be a transit pair. There are $|A_i| + |B_i|$ choices for~$u$, and at most $2c$ choices for~$w$. Therefore, the number of transit pairs over all semi-separated pairs is at most $\sum_{i=1}^k 2c(|A_i| + |B_i|) = O(cp \log p)$, since $\sum_{i=1}^k(|A_i| + |B_i|) = O(p \log p)$ is the weight of the SSPD in~\cite{DBLP:journals/dcg/AbamBFG09}.
\end{proof}

Our next step is to precompute and store the minimum Fr\'echet distance $\frechet(\pi, uw)$ for each transit pair $(u,w)$, where $\pi$ ranges over all paths in~$P$ between~$u$ and~$w$. For this, we use a modification of the algorithm by Alt et al.~\cite{DBLP:journals/jal/AltERW03}.

\begin{lemma}
\label{lemma:straightest_path_computation}
Let $u,w \in P$ be a pair of vertices, and let $ab$ be a segment. One can compute $\min_{\pi} \frechet(\pi,ab)$ in $O(p \log p)$ time, where $\pi$ ranges over all paths in~$P$ between~$u$ and~$w$.
\end{lemma}

\begin{proof}[Proof (Sketch)]
Our proof is essentially the same as in~\cite{DBLP:journals/jal/AltERW03}, except that we replace the sweepline algorithm with a simple Dijkstra search~\cite{DBLP:journals/nm/Dijkstra59}. The fact that the endpoints~$u$ and~$w$ are given makes this simplification possible. Furthermore, by replacing the sweepline with Dijkstra, we do not require~$P$ to be planar.

For the sake of completeness, we provide a proof sketch of our result. We set up a free space diagram for the decision problem in the same way as in~\cite{DBLP:journals/jal/AltERW03}.
Let $P = (V,E)$. For each edge $e \in E$, let~$FD(e,ab)$ be the free space diagram between $e$ and $ab$. Note that the $x$- and $y$-coordinates of~$FD(e,ab)$ denote the positions along $e$ and $ab$ respectively. Moreover, orient~$FD(e,ab)$ so that~$a$ has the minimum $y$-coordinate and~$b$ has the maximum $y$-coordinate. Similarly to~\cite{DBLP:journals/jal/AltERW03}, there is a path $\pi$ between vertices~$u$ and~$w$ satisfying $\frechet(\pi,ab) \leq d$ if and only if there is a sequence of free space diagrams $\{FD(e_i,ab)\}_{i=1}^k$ with a monotone path in the free space from $(u,a)$ to $(w,b)$. See Figure~\ref{fig:stage1_2}.

\begin{figure}[ht]
	\centering
	\includegraphics{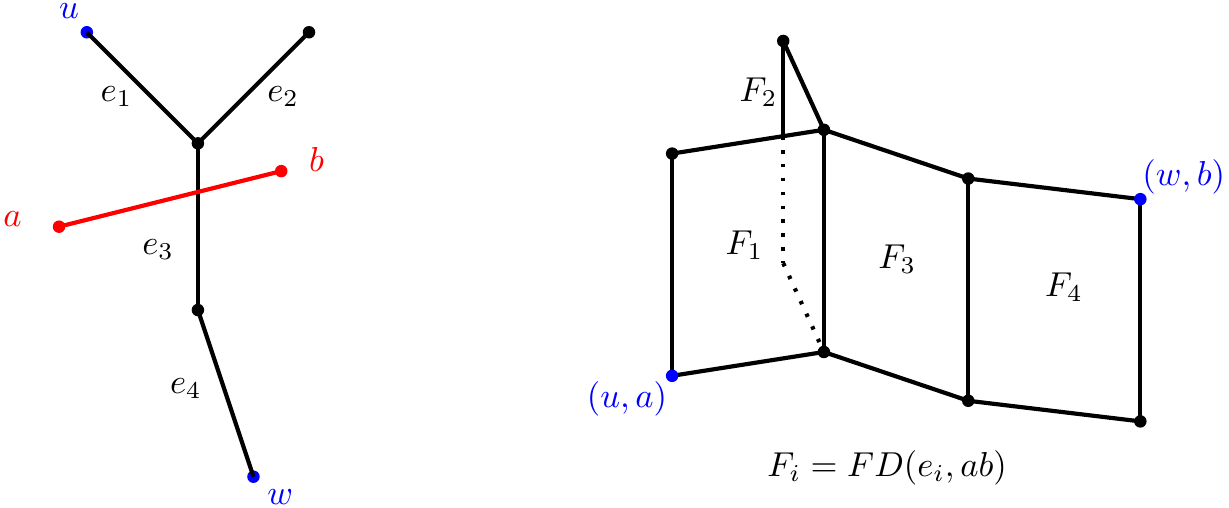}
	\caption{The $x$-coordinates of~$FD(e_i,ab)$ denote the position along $e_i$, and the $y$-coordinates denote the position along $ab$. There is a path $\pi$ between~$u$ and~$w$ with $\frechet(\pi,ab) \leq d$ if and only if there is a monotone path from $(u,a)$ to $(w,b)$ in the free space surface.}
	\label{fig:stage1_2}
\end{figure}

We avoid using a sweepline algorithm, and instead perform a Dijkstra search between~$u$ and~$w$. First, check that $(u,a)$ and $(w,b)$ are in the free space. Next, construct a priority queue on points $(x,y)$ in the free space diagram, where $x$ is a vertex of~$P$, and $y$ is any point on the segment $ab$. The priority of $(x,y)$ is $y$. The invariant maintained by the priority queue is that for all $(x,y)$ in the priority queue, there is a monotone path from $(u,a)$ to $(x,y)$. Initially, the priority queue contains only $(u,a)$. In each iteration, we pop a point from the priority queue with minimum $y$-coordinate. Let this point be $(x,y)$. We mark the point $x$ as visited. For all neighbours $x'$ of $x \in P$, add $(x',y')$ to the priority queue, where $y'$ is the minimum $y$-coordinate so that there is a monotone path from $(x,y)$ to $(x',y')$ in~$FD(xx',ab)$. This maintains the invariant of the priority queue. Halt the priority queue if $(w,b')$ is in the priority queue, which occurs if and only if there is a monotone path from $(u,a)$ to $(w,b')$ to $(w,b)$ (assuming $(w,b)$ is in free space). Finally, we apply parametric search to minimise the Fr\'echet distance in the same way as in~\cite{DBLP:journals/jal/AltERW03}.

We analyse the running time. Constructing the free space diagrams takes $O(p)$ time. Running Dijkstra's algorithm takes $O(|E| + |V| \log |V|) = O(p \log p)$ time. Finally, applying parametric search~\cite{DBLP:journals/jacm/Megiddo83} with Cole's optimisation~\cite{DBLP:journals/jacm/Cole87} takes $O(p \log p)$ time. 
\end{proof}

We are now ready to build the data structure for straightest path queries, which is the main result of this section.

\straightestpathdatastructure*

\begin{proof}
First we describe the preprocessing procedure. Construct an SSPD of the vertices of~$P$, with separation constant $1/2$. For each semi-separated pair $(A_i,B_i)$, let~$C_i$ be its set of transit vertices as defined in Lemma~\ref{lemma:transit_vertices}. Recall that if $u \in A_i \cup B_i$ and $w \in C_i$, then $(u,w)$ is a transit pair of $(A_i,B_i)$. For each transit pair $(u,w)$, we set $ab = uw$ in Lemma~\ref{lemma:straightest_path_computation} to compute the straightest path between~$u$ and~$w$, and we store the minimum Fr\'echet distance. 

Next, we describe the query procedure. Given a pair of query vertices~$u$ and~$v$, we query our SSPD for the semi-separated pair $(A_i, B_i)$ so~$(i)$ $u \in A_i$ and $v \in B_i$, or~$(ii)$ $v \in A_i$ and $u \in B_i$. Let~$C_i$ be the transit vertices for $(A_i,B_i)$. For each $w \in C_i$, define $\pi_{uw}$ to be the straightest path between~$u$ and~$w$, and define $D_{uw} = \frechet(\pi_{uw},uw)$. Define $\pi_{wv}$ and $D_{wv}$ analogously. Define $t$ to be the orthogonal projection of~$w$ onto $uv$ and define $D_w$ to be the orthogonal distance. See Figure~\ref{fig:stage1_3}. Finally, return $\min_{w \in C_i} (\max(D_{uw}, D_{wv}) + D_w)$ as a 3-approximation for the minimum Fr\'echet distance of the shortest path between~$u$ and~$v$.

\begin{figure}[ht]
	\centering
	\includegraphics{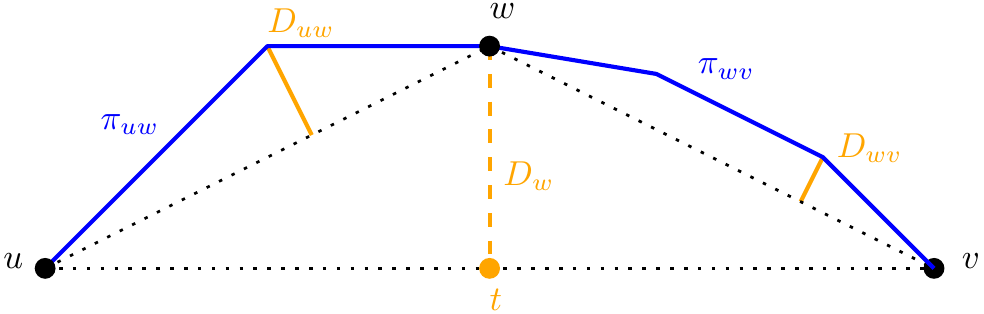}
	\caption{Vertices $u,v$ and transit vertex~$w$ (black), the straightest paths $\pi_{uw}$ and $\pi_{wv}$ (blue) with Fr\'echet distances $D_{uw}, D_{wv}$ (orange), and the orthogonal distance $D_w$ from~$w$ to $uv$ (orange, dashed).}
	\label{fig:stage1_3}
\end{figure}

We prove that the query procedure returns a 3-approximation. Our proof is inspired by the proof of Lemma~5.5 in~\cite{DBLP:journals/siamcomp/DriemelH13}. Let $\pi_{uv}$ be the straightest path between~$u$ and~$v$. Then, for any transit vertex $w \in C_i$, we have

\[
    \begin{array}{rcl}
        \frechet(\pi_{uv}, uv) 
        &\leq& \frechet(\pi_{uw}\circ \pi_{wv}, uv) \\
        &\leq& \max(\frechet(\pi_{uw}, ut), \frechet(\pi_{wv}, tv))) \\
        &\leq& \max(\frechet(\pi_{uw}, uw) + \frechet(uw,ut), \frechet(\pi_{wv}, wv) + \frechet(wv,tv)) \\
        &=& \max(D_{uw} + D_w, D_{wv} + D_w) \\
        &=& \max(D_{uw},D_{wv}) + D_w \\
    \end{array}
\]
Therefore, $\frechet(\pi_{uv}, uv) \leq \min_{w \in C_i} (\max(D_{uw}, D_{wv}) + D_w)$. Next, using Lemma~\ref{lemma:transit_vertices}, assume that $w^* \in C_i$ is a transit vertex so that $w^* \in \pi_{uv}$. Then clearly $D_{w^*} \leq \frechet(\pi_{uv}, uv)$. Next, we will use Lemma~5.3 in~\cite{DBLP:journals/siamcomp/DriemelH13}, which states that for any subcurve $Z'$ of $Z$ we have $d_F(spine(Z'),Z') \leq  2 d_F(spine(Z),Z))$, where $spine(X)$ denotes the segment joining the endpoints of curve $X$. Applying this lemma to $Z = \pi_{uv}$ and $Z' = \pi_{uw^*}$, we have $D_{uw^*} \leq 2 \cdot \frechet(\pi_{uv}, uv)$. Putting this together, we obtain $\max(D_{uw^*}, D_{w^*v}) + D_w^* \leq 3 \cdot \frechet(\pi_{uv}, uv)$. Therefore, our query procedure returns a 3-approximation of $\frechet(\pi_{uv}, uv)$, as required.

Finally, we analyse the running time and space of our preprocessing and query procedures. Constructing the SSPD takes $O(p \log p)$ time~\cite{DBLP:journals/dcg/AbamBFG09}. By Lemma~\ref{lemma:transit_vertices}, all transit vertices and transit pairs can be computed in $O(c p^2 \log p)$ time. Computing the minimum Fr\'echet distance for all transit pairs takes $O(cp^2 \log^2 p)$ time. Therefore, our data structure can be constructed in $O(cp^2 \log^2 p)$ time. Storing the Fr\'echet distance for all transit pairs requires $O(cp \log p)$ space, by Lemma~\ref{lemma:number_of_transit_vertices_and_transit_pairs}. By Observation~\ref{obs:query_sspd}, querying the SSPD for the semi-separated pair containing the query vertices takes $O(\log p)$ time. There are $O(c)$ transit vertices to check. For a transit vertex~$w$, looking up the values $D_{uw}$ and $D_{wv}$ in our data structure takes constant time. Computing $D_w$ takes constant time.  Putting this all together, we obtain the stated theorem.
\end{proof}

\section{Stage 2: Map matching segment queries}
\label{sec:map_matching_segment_queries}

Recall that a data structure for map matching segment queries is defined as follows. Given a query segment $ab$ in the plane, the data structure returns the minimum Fr\'echet distance $\frechet(\pi,ab)$ as~$\pi$ ranges over all paths in~$P$ that start and end at a vertex of~$P$. 

As stated in the technical overview, we build two data structures in this section. The first data structure in this section is an extension of Theorem~\ref{theorem:straightestpathdatastructure}, which we modify to handle arbitrary query segments in the plane. 
\begin{lemma}
\label{lemma:straightest_path_query_but_with_ab}
Given a $c$-packed graph $P$ of complexity $p$, one can construct a data structure of $O(cp \log p)$ size, so that given a pair of query vertices $u,v \in P$ and a query segment $ab$ in the plane, the data structure returns in $O(\log p)$ query time a $3$-approximation of $\min_\pi \frechet(\pi,ab)$, where $\pi$ ranges over all paths in $P$ between $u$ and $v$. The preprocessing time is $O(cp^2 \log^2 p)$.
\end{lemma}

\begin{proof}[Proof (Sketch)]
Our proof is the exactly same as the proof of Theorem~\ref{theorem:straightestpathdatastructure}, except that $(i)$ we define $D_w = \frechet(uw \circ wv, ab)$, and $(ii)$ we define $t$ to be the point on $ab$ that matches to $w$, under the minimum Fr\'echet distance matching between $ab$ and $uw \circ wv$.
\end{proof}

This completes the construction of the first data structure, however, its approximation ratio is~$3$. Next, we improve the approximation ratio to $(1+\varepsilon)$. The next lemma is inspired by Lemma~5.8 in~\cite{DBLP:journals/siamcomp/DriemelH13}, which uses $\log(1/\varepsilon)$ number of $\varepsilon$-grids to ensure that at least one grid size is within a factor of $(1+\varepsilon)$ of the true Fr\'echet distance.

\begin{lemma}
\label{lemma:exponential_grid_between_transit_pairs}
Let $u,w \in P$ be a fixed pair of vertices. Let $\varepsilon > 0$ and $\chi = \varepsilon^{-2} \log (1/\varepsilon)$. One can construct a data structure of $O(\chi^2)$ space, so that given a query segment $ab$ in the plane, the data structure returns in constant time a $(1+\varepsilon)$-approximation of $\min_{\pi} \frechet(\pi, ab)$, where $\pi$ ranges over all paths in~$P$ between~$u$ and~$w$. The preprocessing time is $O(\chi^2 p \log p)$. 
\end{lemma}

\begin{proof}[Proof (Sketch)]
Our proof is the same as the proof of Lemma~5.8 in~\cite{DBLP:journals/siamcomp/DriemelH13}, except that instead of computing the Fr\'echet distance between a curve and a segment joining a pair of grid points, we use Lemma~\ref{lemma:straightest_path_computation} to minimise the Fr\'echet distance over all paths between~$u$ and~$w$.
\end{proof}

Now, we use Lemma~\ref{lemma:exponential_grid_between_transit_pairs} to improve the approximation ratio of Lemma~\ref{lemma:straightest_path_query_but_with_ab} to $(1+\varepsilon)$.

\begin{lemma}
\label{lemma:map_matching_segment_query_if_u_and_v_are_known}
Let $\varepsilon >0$ and $\chi = \varepsilon^{-2} \log(1/\varepsilon)$. One can construct a data structure of $O(c \chi^2 p \log p)$ size, so that given a pair of query vertices $u,v \in P$ and a query segment $ab$ in the plane, the data structure returns in $O(\log p + c \varepsilon^{-1})$ time a $(1+\varepsilon)$-approximation of $\min_{\pi} \frechet(\pi,ab)$, where $\pi$ ranges over all paths in~$P$ between~$u$ and~$v$. The preprocessing time is $O(c \chi^2 p^2 \log^2 p)$. 
\end{lemma}

\begin{proof}[Proof (Sketch)]
Our proof is essentially the same as the proof of Theorem~5.9 in~\cite{DBLP:journals/siamcomp/DriemelH13}, except that~$(i)$ we replace subcurves with transit pairs,~$(ii)$ we replace Lemma~5.8 in~\cite{DBLP:journals/siamcomp/DriemelH13} with Lemma~\ref{lemma:exponential_grid_between_transit_pairs}, and $(iii)$ we replace Theorem~5.6 in~\cite{DBLP:journals/siamcomp/DriemelH13} with Lemma~\ref{lemma:straightest_path_query_but_with_ab}.

For the sake of completeness, we provide a proof sketch of our result. We first describe the preprocessing procedure. We construct the data structure in Lemma~\ref{lemma:straightest_path_query_but_with_ab}. For each transit pair, we construct the data structure in Lemma~\ref{lemma:exponential_grid_between_transit_pairs}. Next, we describe the query procedure. Given a pair of query vertices $(u,v)$ and a query segment $ab$, we use Lemma~\ref{lemma:straightest_path_query_but_with_ab} to compute a real value $r$ so that $\min_{\pi} \frechet(\pi,ab) \leq r \leq 3 \cdot \min_{\pi} \frechet(\pi,ab)$, where $\pi$ ranges over all paths between~$u$ and~$v$. Next, we iterate over all transit vertices $w$ associated with the semi-separated pair containing $(u,v)$. Define~$B(w,3r)$ to be a ball with radius $3r$ centred at~$w$. If $B(w,3r)$ does not intersect~$ab$, we skip the transit vertex $w$ and move onto the next one. Hence, we may assume that $B(w,3r)$ intersects~$ab$. We compute $O(\varepsilon^{-1})$ evenly spaced vertices on the chord $B(w,3r) \cap ab$. Let~$t$ be one of these vertices. See Figure~\ref{fig:stage2_1}.

\begin{figure}[ht]
	\centering
	\includegraphics{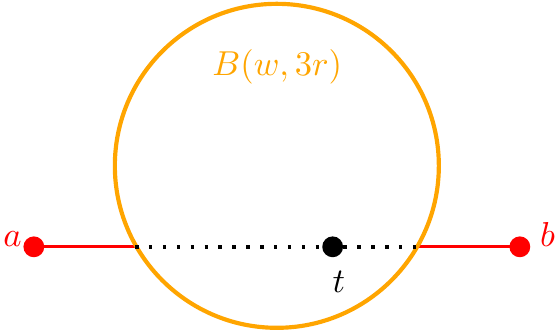}
	\caption{The $O(\varepsilon^{-1})$ evenly spaced vertices, including $t$ (black) on the chord $B(w,3r) \cap ab$ which is the intersection of segment $ab$ (red) and the ball centred at~$w$ with radius $3r$ (orange).}
	\label{fig:stage2_1}
\end{figure}

 We use Lemma~\ref{lemma:exponential_grid_between_transit_pairs} to compute a $(1+\varepsilon)$-approximation of $\min_{\pi} \frechet(\pi, at)$ as~$\pi$ ranges over all paths between~$u$ and~$w$ (resp. $tb$ and paths between~$w$ and~$v$). We take the larger Fr\'echet distance out of the $at$ and $tb$ cases, and return it as a $(1+\varepsilon)$-approximation of $\min_{\pi} \frechet(\pi, ab)$ as~$\pi$ ranges over all paths between~$u$ and~$v$ assuming $w \in \pi$ and~$w$ matches to $t$. Finally, we minimise over all transit vertices~$w$ and the evenly spaced vertices $t \in B(w,3r) \cap ab$ to obtain a $(1+\varepsilon)$-approximation of $\min_{\pi} \frechet(\pi, ab)$. The proof of correctness follows from Theorem~5.9 in~\cite{DBLP:journals/siamcomp/DriemelH13}. 

We analyse the preprocessing time and space. There are $O(cp \log p)$ transit pairs by Lemma~\ref{lemma:number_of_transit_vertices_and_transit_pairs}. By Lemma~\ref{lemma:straightest_path_query_but_with_ab} and Lemma~\ref{lemma:exponential_grid_between_transit_pairs}, the data structure has $O(c \chi^2 p \log p)$ size, and can be constructed in $O(c\chi^2 p^2 \log^2)$ preprocessing time. We analyse the query time. Computing a 3-approximation takes $O(\log p)$ query time. Iterating over all choices of~$w$ and $t$ takes $O(c \varepsilon^{-1})$ time. Putting this together yields the claimed lemma.
\end{proof}

This improves the approximation ratio of the first data structure to~$(1+\varepsilon)$. Next, we consider the second data structure, which can efficiently query the starting and ending points of the path. 

For our second data structure, we simplify the $c$-packed graph using graph clustering. The clustering algorithm we use is Gonzales' algorithm~\cite{DBLP:journals/tcs/Gonzalez85}. Let $P = (V,E)$ be the graph, which from Section~\ref{sec:preliminaries} is assumed to be connected. For a pair of vertices $u,v \in V$, let $\graphdist(u,v)$ be the shortest path between~$u$ and~$v$ in~$P$. For $k=1,\ldots,p$, we compute a $k$-centre clustering of~$V$ under the graph metric $\graphdist$. For $k=1$, choose an arbitrary vertex~$v_1$ to be the $1$-centre. Mark~$v_1$ as a cluster centre, and let~$r_1$ be the radius of the 1-centre clustering. For $k \geq 2$, compute the vertex~$v_k$ that is the furthest from all existing cluster centres~$v_1,\ldots,v_{k-1}$. Mark $v_k$ as a new  centre, and let~$r_k$ be the radius of the $k$-centre clustering. After all vertices are marked as cluster centres, we have computed a list $[(v_1,r_1),\ldots,(v_p,r_p)]$ of cluster centres and cluster radii.

We use the cluster centres and cluster radii to construct a hierarchy of simplifications of the graph~$P$. In particular, define~$V_r$ to be the set of vertices $\{v_i \in V: r_i \geq \varepsilon r\}$. We show that for any square $S$ with side length~$2r$, there are at most $O(c \varepsilon^{-1})$ vertices in $V_r \cap S$.

\begin{lemma}
\label{lemma:points_in_square_lemma}
Let $P = (V,E)$ be a $c$-packed graph and let $S$ be a square in the plane with side length~$2r$. Then there exists a set of vertices $T \subseteq V$ satisfying~$(i)$ $|T| = O(c\varepsilon^{-1})$ and~$(ii)$ for all vertices $v \in V \cap S$, there exists $t \in T$ so that $\graphdist(v,t) \leq \varepsilon r$.
\end{lemma}

\begin{proof}
Run the clustering algorithm described above to compute a list $[(v_1,r_1),\ldots,(v_p,r_p)]$ of cluster centres and their clustering radii. Recall that $V_r$ is the set of vertices in~$V$ satisfying~$r_i \geq \varepsilon r$. Let $S'$ be a square concentric with $S$, but has side length $4r$. Define $T_1 = V_r \cap S'$. First, we show that $T_1$ satisfies Property~$(i)$. Then we add a single vertex to $T_1$ to construct $T_2$, and show that $T_2$ satisfies both Properties~$(i)$ and~$(ii)$.

For Property~$(i)$, if there exists a vertex $t \in T_1$ so that $\graphdist(t,t') \leq \varepsilon r$ for all $t' \in V$, then defining $T_1 = \{t\}$ clearly satisfies both properties. Otherwise, for all $t \in T_1$ pick a vertex $t'$ so that $\graphdist(t,t') \geq \varepsilon r$. Construct the shortest path from~$t$ to~$t'$ under the graph metric $\graphdist$, and let $t''$ be the point (not necessarily a vertex) on the shortest path between~$t$ and~$t'$ so that $\graphdist(t,t'') = \varepsilon r /3$. Let the shortest path from~$t$ to~$t''$ be~$\pi_t$. Construct the set of paths $\{\pi_t\}_{t \in T_1}$. First, we will show that the set of paths $\{\pi_t\}_{t \in T_1}$ are edge disjoint, and all lie in a square with side length $5r$. See Figure~\ref{fig:stage2_2}. Then, we will use the $c$-packedness property in the square with side length $5r$ to prove Property~$(i)$.

\begin{figure}[ht]
	\centering
	\includegraphics{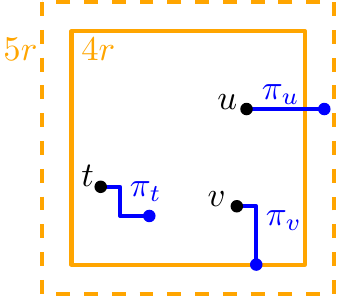}
	\caption{The vertices $T_1 = \{t,u,v\}$ (black) are in a square of side length $4r$ (orange, solid). The paths $\{\pi_t,\pi_u,\pi_v\}$ (blue) are edge disjoint and lie in a square with side length $5r$ (orange, dashed).}
	\label{fig:stage2_2}
\end{figure}

First, we show $\{\pi_t\}_{t \in T_1}$ is edge disjoint. Suppose for the sake of contradiction that $s,t \in T_1$, and $\pi_s$ and $\pi_t$ share an edge. Using this shared edge, by the triangle inequality we have $\graphdist(s,t) < \varepsilon r$. Let $s = v_i$ have cluster radius~$r_i$, and $t = v_j$ have cluster radius~$r_j$, so that $i<j$. Then we have the inequality
\[
    r_j \leq r_{j-1} = \max_{v \in P} \min_{k <j} \graphdist(v_k,v) = \min_{k<j} \graphdist(v_k,v_j) \leq \graphdist(v_i,v_j) = \graphdist(s,t) < \varepsilon r,
\]
where $\max_{v \in P} \min_{k <j} \graphdist(v_k,v) = \min_{k<j} \graphdist(v_k,v_j)$ comes from the fact that $v_j$ was the furthest vertex from all existing cluster centres in the $j^{th}$ round of Gonzales' algorithm. But $t \in V_r$, so~$r_j \geq \varepsilon r$. This is a contradiction, so $\{\pi_t\}_{t \in T_1}$ is edge disjoint. 

Now, we will use the $c$-packedness property to show that $|T_1| = O(c\varepsilon^{-1})$. Each path $\pi_t$ is a shortest path between a pair of points that are distance $\varepsilon r /3$ away, so the total path length of $\{\pi_t\}_{t \in T_1}$ is $|T_1| \cdot \varepsilon r /3$. Each vertex $t \in T_1$ is in a square with side length $4r$. Each path $\pi_t$ has length at most~$\varepsilon r < r$, since $0 < \varepsilon < 1$. Therefore, all edges in $\{\pi_t\}_{t \in T_1}$ are in a square with side length $5r$. Finally, by $c$-packedness, we have that $c \cdot 5r$ is an upper bound on the total edge length in the square of side length $5r$, which is guaranteed to contain the edges of $\{\pi_t\}_{t \in T_1}$. Therefore, $c \cdot 5r \geq |T_1| \cdot \varepsilon r /3$, which implies $|T_1| \leq 15c\varepsilon^{-1}$. This concludes the proof of Property~$(i)$.

For Property~$(ii)$, let $V_r = [v_1, \ldots, v_i]$. If $V_r$ consists of all vertices of~$V$, then $T_1$ consists of all vertices inside $S'$, and~$(ii)$ is trivially true. Otherwise, suppose $v_{i+1} \not \in V_r$. So~$r_{i+1} < \varepsilon r$. Therefore, after $i+1$ rounds of Gonzales' algorithm, the cluster radius is at most $\varepsilon r$. So all vertices $v \in P$ are within distance $\varepsilon r$ from one of the vertices $[v_1,\ldots, v_{i+1}]$. Define $T_2 = [v_1, \ldots, v_{i+1}] \cap S'$. Then $|T_2| \leq |T_1| + 1 = O(c \varepsilon^{-1})$. Moreover, for all vertices $v \in P \cap S$,~$v$ is within distance $\varepsilon r$ from one of the vertices $[v_1,\ldots, v_{i+1}]$. Without loss of generality, let $j \leq i+1$ so that $\graphdist(v,v_j) < \varepsilon r$. Since $v \in S$ and $\graphdist(v,v_j) < \varepsilon r$, we have $v_j \in S'$. Therefore, $v_j \in T_2 = [v_1,\ldots, v_{i+1}] \cap S'$. As a result, $T_2$ satisfies both Properties~$(i)$ and~$(ii)$, as required. 
\end{proof}

Next, we build a data structure so that, given any square $S$ in the plane, the data structure can efficiently return a set of vertices $T$ that satisfies the properties in Lemma~\ref{lemma:points_in_square_lemma}. 

\begin{lemma}
\label{lemma:starting_and_ending_point_data_structure}
Let $P = (V,E)$, and let $\varepsilon > 0$. One can construct a data structure of $O(p \log p)$ size, so that given a query square $S$ in the plane with side length $2r$, the data structure returns in $O(\log p + c \varepsilon^{-1})$ time a set of vertices $T$ satisfying~$(i)$ $|T| = O(c \varepsilon^{-1})$ and~$(ii)$ for all vertices $v \in V \cap S$, there exists $t \in T$ so that $\graphdist(v,t) \leq \varepsilon r$. The preprocessing time is $O(p^2 \log p)$. 
\end{lemma}

\begin{proof}
We show how to efficiently query the set $T_2$ given in Lemma~\ref{lemma:points_in_square_lemma}. We run the clustering algorithm described in this section to compute a list $[(v_1, r_1), \ldots, (v_p, r_p)]$ of cluster centres and their clustering radii. We build an orthogonal range searching data structure for three-dimensional points. Specifically, for each pair $(v_i,r_i)$ in the list, we insert the point $(x_i, y_i, r_i)$ into the orthogonal range searching data structure, where $(x_i,y_i)$ are the $x$- and $y$-coordinates of the point $v_i$, respectively. Given a square $S$, we perform an orthogonal range search for all vertices $(v_i,r_i)$ so that $v_i \in S$, and~$r_i \geq \varepsilon r$. We return this set of vertices as $T_2$. Lemma~\ref{lemma:points_in_square_lemma} proves that $T_2$ satisfies Properties~$(i)$ and~$(ii)$. 

Next, we analyse the running time and space of the preprocessing and query procedures. We use the orthogonal range searching data structure of Afshani et~al.~\cite{DBLP:conf/compgeom/AfshaniAL10}, and we use the clustering algorithm of Gonzales~\cite{DBLP:journals/tcs/Gonzalez85}.

The storage requirement of the orthogonal range searching data structure is $O(p \log^2 p)$. Computing the distance matrix under $\graphdist$ takes $O(p^2 \log p)$ time. Performing Gonzales' clustering algorithm to compute~$p$ centres takes $O(p^2)$ time. Constructing the orthogonal range searching data structure takes $O(p \log^2 p)$ time. Therefore, the overall preprocessing time is $O(p^2 \log p)$. The query time of the orthogonal range searching data structure is $O(\log p + |T|)$. Therefore, the overall query time is $O(\log p + c \varepsilon^{-1})$. This proves the stated lemma.
\end{proof}

Finally, we are ready to combine the two data structures in Lemmas~\ref{lemma:map_matching_segment_query_if_u_and_v_are_known} and~\ref{lemma:starting_and_ending_point_data_structure} to answer map matching segment queries. The theorem below is the main result of this section.

\mapmatchingsegmentqueries*

\begin{proof}
The preprocessing procedure is to construct the data structure in Lemmas~\ref{lemma:map_matching_segment_query_if_u_and_v_are_known} and~\ref{lemma:starting_and_ending_point_data_structure}. For both data structures, we use the parameter $\varepsilon' = \varepsilon/6$ instead of~$\varepsilon$. 

Next, we consider the query procedure for the decision problem. Given a segment~$ab$ in the plane and a Fr\'echet distance of~$r$, the decision problem is to decide whether $r^* \leq r$ or $r^* \geq r$, where $r^* = \min_{\pi} \frechet(\pi,ab)$ where $\pi$ ranges over all paths in~$P$ that start and end at a vertex of~$P$. We construct a disk centred at~$a$ with radius~$r$, and we enclose this disk in a square with side length~$2r$. We query the data structure in Lemma~\ref{lemma:starting_and_ending_point_data_structure} to obtain a set of vertices $T_a$. We query the set of vertices $T_b$ analogously. For every $(u,v) \in T_a \times T_b$, we query the data structure in Lemma~\ref{lemma:map_matching_segment_query_if_u_and_v_are_known} for a value~$r_{uv}$ which is a $(1+\varepsilon')$-approximation of $\min_{\pi} \frechet(\pi,ab)$, where $\pi$ ranges over all paths between~$u$ and~$v$. See Figure~\ref{fig:stage2_3}.

\begin{figure}[ht]
	\centering
	\includegraphics{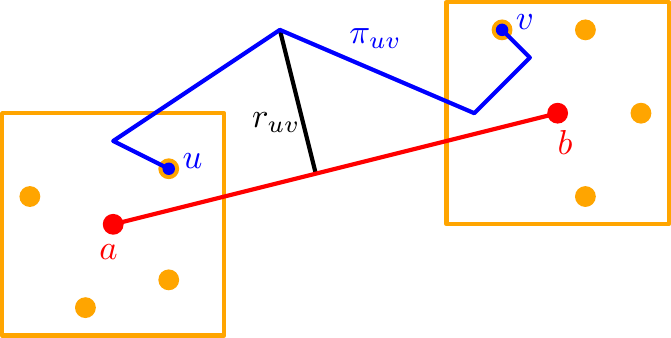}
	\caption{The query segment $ab$ (red), the squares centred at~$a$ and~$b$ with side length $2r$ (orange), the sets of vertices $T_a$ and $T_b$ including~$u$ and~$v$ (orange), and the path $\pi_{uv}$ (blue) between~$u$ and~$v$ minimising the Fr\'echet distance~$r_{uv}$ (black) up to a factor of $(1+\varepsilon)$.}
	\label{fig:stage2_3}
\end{figure}

Let $r' = \min_{(u,v) \in T_a \times T_b} r_{uv}$.  We distinguish three cases (a), (b) and (c):

\begin{enumerate}[label=(\alph*), noitemsep]
    \item If $r' \leq r$, we return that $r^* \leq r$.
    \item If $r' \geq (1+\varepsilon')^2 r$, we return that $r^* \geq r$.
    \item If $r' \in [r, (1+\varepsilon')^2 r]$, we return that $r^* \in [(1-\varepsilon') r, (1+\varepsilon')^2 r]$. 
\end{enumerate}
The third case does not technically answer the decision problem, as it does not return $r^* \leq r$ or $r^* \geq r$. However, in this case, we will show that $(1+\varepsilon')r$ is a $(1+\varepsilon)$-approximation of $r^*$, as required by the stated theorem. This completes the description of the query procedure for the decision problem. Next, we prove its correctness, which we separate into cases (a), (b) and (c). 
\begin{enumerate}[label=(\alph*), noitemsep]
    \item We know $r^* \leq r_{uv}$ for all vertices $u,v$ in~$P$. Therefore, $r^* \leq r'$. So if $r'\leq r$, then $r^* \leq r$.
    \item Given a pair of vertices $u, v$ in~$P$, define $r^*_{uv}$ to be $\min_{\pi} \frechet(\pi,ab)$ where $\pi$ ranges over all paths between~$u$ and~$v$. Clearly, $r^* \leq r^*_{uv}$ for all vertices~$u,v$ in~$P$. Moreover, by Lemma~\ref{lemma:map_matching_segment_query_if_u_and_v_are_known}, since~$r_{uv}$ is a $(1+\varepsilon')$-approximation, we have $r^* \leq r^*_{uv} \leq r_{uv} \leq (1+\varepsilon') r_{uv}^*$. Let $\pi^*$ be the path that attains~$r^*$, i.e. $r^* = \frechet(\pi^*,ab)$. Let the starting and ending points of $\pi^*$ be $u^*$ and $v^*$ respectively. By Lemma~\ref{lemma:points_in_square_lemma}, there exists a graph vertex $u \in T_a$ so that $\graphdist(u^*,u) \leq \varepsilon' r$. Define the vertex $v \in T_b$ analogously. See Figure~\ref{fig:stage2_4}. Consider the path $\pi_{uv}^*$ obtained by concatenating the paths $uu^*$, $\pi^*$, and $v^*v$. Note that $\pi_{uv}^*$ is a valid path between~$u$ and~$v$, moreover, $\frechet(\pi_{uv}^*,ab) \leq \max(\frechet(uu^*, a), \frechet(\pi^*,ab), \frechet(v^*v,b))$. But $\frechet(u^*,a) \leq \frechet(\pi^*,ab) = r^*$, and $\graphdist(u^*,u) \leq \varepsilon' r$. Hence, $\frechet(uu^*, a) \leq r^* + \varepsilon' r$. Similarly, $\frechet(v^*v, b) \leq r^* + \varepsilon' r$, so $\frechet(\pi_{uv}^*, ab) \leq r^* + \varepsilon' r$. Therefore,~$r_{uv}^* \leq \frechet(\pi_{uv}^*, ab) \leq r^* + \varepsilon' r$. Now, $r' \leq r_{uv} \leq (1+\varepsilon')r^*_{uv} \leq (1+\varepsilon')(r^* + \varepsilon' r)$. If $r' \geq (1+\varepsilon')^2 r$, then $(1+\varepsilon')(r^* + \varepsilon' r) \geq (1+\varepsilon')^2 r$, so $r^* + \varepsilon' r \geq r + \varepsilon' r$, and therefore $r^* \geq r$. 
    
    \item From the proof of the first case, we have $r^* \leq r'$. If $r' \leq (1+\varepsilon')^2 r$, then $r^* \leq (1+\varepsilon')^2 r$. From the proof of the second case, we have $r' \leq r^* + \varepsilon' r$. If $r' \geq r$, then $r^* \geq (1-\varepsilon')r$. Putting this together, if $r' \in [r, (1+\varepsilon')^2 r]$, then $r^* \in [(1-\varepsilon') r, (1+\varepsilon')^2 r]$. In particular, $(1+\varepsilon')^2 r \geq r^*$, and $(1+\varepsilon')^2 r \leq (1+\varepsilon')^3 r^* \leq (1+\varepsilon')(1+3 \varepsilon') r^* \leq (1+6\varepsilon') r^* = (1+\varepsilon) r^*$, so $(1+\varepsilon') r$ is a $(1+\varepsilon)$-approximation of $r^*$, as required.
\end{enumerate}

    \begin{figure}[ht]
    	\centering
    	\includegraphics{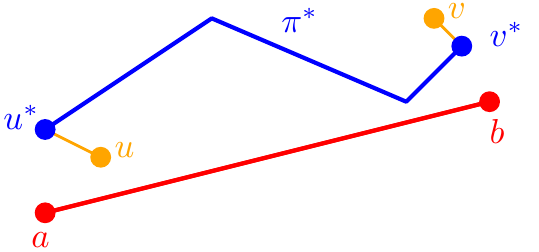}
    	\caption{Given segment $ab$ (red), the path $\pi^*$ (blue) minimises its Fr\'echet distance to $ab$. We define a new path that concatenates vertex~$u$, path $\pi^*$, and vertex~$v$, where $u \in T_a$ and $v \in T_b$.}
    	\label{fig:stage2_4}
    \end{figure}

Next, we apply parametric search to the decision problem, which we call $D(r)$. Define~$D(r)$ to be TRUE if $r' \leq r$, and define~$D(r)$ to be FALSE if $r' \geq (1+\varepsilon')r$. If $r' \in [r, (1+\varepsilon')r]$, we immediately halt the parametric search, and return $(1+\varepsilon')r$ as a $(1+\varepsilon)$-approximation of $r^*$. It suffices to show~$(i)$ that $D(r)$ is monotone, and~$(ii)$ that all operations in $D(r)$ are either independent of $r$, or can be made equivalent to a constant number of comparisons $\{r > c_i\}$ where~$c_i$ is a critical value. First, we show~$(i)$. Suppose $D(r_1)$ evaluates to TRUE, and~$r_1 < r_2$. Then $r^* \leq r_1 \leq r_2$, and we cannot have $D(r_2)$ evaluating to FALSE. So either~$D(r_2)$ is also TRUE, or we halt the parametric search and obtain a $(1+\varepsilon)$-approximation. Similarly, if $D(r_1)$ evaluates to FALSE, and~$r_1 > r_2$, then we cannot have $D(r_2)$ evaluating to TRUE. Therefore, $D(r)$ is monotone. Next, we show~$(ii)$. The first step of $D(r)$ is to query the data structure in Lemma~\ref{lemma:starting_and_ending_point_data_structure} for the set $T_a$. The data structure is a three-dimensional orthogonal range searching data structure. All operations that depend on~$r$ can be evaluated by comparing $r$ to a difference between $x$-, $y$- or $z$-coordinates. As an example, if $a = (x_a,y_a)$, then a point $(x_i, y_i, r_i)$ lies in the orthogonal range if and only if $|x_a - x_i| \leq r$, $|y_a - y_i| \leq r$ and~$r_i \geq \varepsilon' r$. In particular, $|x_a - x_i|$, $|y_a - y_i|$ and~$r_i/\varepsilon'$ are critical values. The next step is to query Lemma~\ref{lemma:map_matching_segment_query_if_u_and_v_are_known} to obtain~$r'$. This step is independent of~$r$ and generates no critical values. Finally, we compare $r'$ with $r$ and $(1+\varepsilon')r$. Therefore, $r'$ and~$r'/(1+\varepsilon')$ are critical values. This completes the proof of~$(i)$ and~$(ii)$.

We analyse the construction time and space of the data structure. Let $\chi = \varepsilon^{-2} \log(1/\varepsilon)$. The data structure in Lemma~\ref{lemma:map_matching_segment_query_if_u_and_v_are_known} has $O(c \chi^2 p \log p)$ size, whereas the data structure in Lemma~\ref{lemma:starting_and_ending_point_data_structure} has $O(p \log p)$ size. The data structure in Lemma~\ref{lemma:map_matching_segment_query_if_u_and_v_are_known} requires $O(c\chi^2 p^2 \log^2 p)$ preprocessing time, whereas the data structure in Lemma~\ref{lemma:starting_and_ending_point_data_structure} requires $O(p^2 \log p)$ preprocessing time. Therefore, the overall data structure has $O(c \chi^2 p \log p)$ size, and requires $O(c\chi^2 p^2 \log^2 p)$ preprocessing time.

We analyse the query time of the decision problem. Querying the data structure in Lemma~\ref{lemma:starting_and_ending_point_data_structure} takes $O(\log p + c\varepsilon^{-1})$ time. There are $O(c^2 \varepsilon^{-2})$ pairs $(u,v) \in T_a \times T_b$. For each pair $(u,v)$, querying the data structure in Lemma~\ref{lemma:map_matching_segment_query_if_u_and_v_are_known} takes $O(\log p)$ time. Therefore, the overall query time of the decision version is $O(c^2 \varepsilon^{-2} \log p)$. We analyse the running time of parametric search. The decision version forms both the sequential and parallel algorithms. The running time of parametric search~\cite{DBLP:journals/jacm/Megiddo83} is $O(N_p T_p + T_p T_s \log N_p)$, where $T_s$ is the running time of the sequential algorithm, $N_p$ is the number of processors for the parallel algorithm, and $T_p$ is the number of parallel steps in the parallel algorithm. By setting $N_p=1$, we obtain $T_s = T_p = O(c^2 \varepsilon^{-2} \log p)$. Therefore, the overall running time of the parametric search step is $O(c^4 \varepsilon^{-4} \log^2 p)$. 
\end{proof}

\section{Stage 3: Map matching queries}
\label{sec:map_matching_queries}

We start by considering the decision problem of the map matching query, in which we are given a trajectory~$Q$ in the plane and a Fr\'echet distance~$r$, and we are to decide whether $r \leq r^*$ or $r \geq r^*$, where $r^*$ is the minimum value of $\frechet(\pi,Q)$ where $\pi$ ranges over all paths in~$P$ that start and end at a vertex of~$P$. Let~$Q$ have vertices $a_1, \ldots, a_q$. The first step is to compute a constant number of points on~$P$ that can match to $a_i$. We have two cases, either the point matching to $a_i$ is a vertex of~$P$, or it is a point along an edge of~$P$. The points along the edges of~$P$ are defined as follows.

\begin{definition}
Given a graph $P = (V,E)$ embedded in the Euclidean plane, define the set~$F$ to be points that lie on an edge of $E$. Formally, $F = \{f \in \mathbb R^2: f \in e, e \in E\}$. 
\end{definition}

The trajectory vertex $a_i$ must match to a point in~$V$ or~$F$. If the point is in~$V$, we use Lemma~\ref{lemma:points_in_square_lemma} to compute a set of $O(c \varepsilon^{-1})$ vertices that can match to~$a_i$. If the point is in~$F$, we prove a generalisation of Lemma~\ref{lemma:points_in_square_lemma} to compute a set of~$O(c \varepsilon^{-2})$ points that can match to $a_i$. The generalisation is stated below.

\begin{lemma}
\label{lemma:points_in_square_lemma_version_2}
Let $P = (V,E)$ and let $F = \{f \in \mathbb R^2: f \in e, e \in E\}$. Let $S$ be a square in the plane with side length $2r$. Then there exists a set of points $T \subset F$ satisfying~$(i)$ $|T| = O(c \varepsilon^{-2})$ and~$(ii)$ for all points $f \in F \cap S$, there exists $t \in T$ so that $\graphdist(f,t) \leq \varepsilon r$.
\end{lemma}

\begin{proof}
We use Lemma~\ref{lemma:points_in_square_lemma} to construct a set of graph vertices $T_2$ so that $|T_2| = O(c \varepsilon^{-1})$, and for all vertices $v \in V \cap S$, there exists $t_2 \in T_2$ so that $\graphdist(v,t) \leq \varepsilon r/2$. Let $E_r$ be the set of edges with length at least $\varepsilon r/2$. Let $S'$ be a square that is concentric with~$S$, but has side length $4r$. For each $e \in E_r$, choose $O(\varepsilon^{-1})$ evenly spaced points on the chord $e \cap S'$, so that the distance between consecutive points is at most $\varepsilon r /2$. Add these $O(\varepsilon^{-1})$ points to the set $T_3$, for each $e \in E_r$. We will show that the set $T_2 \cup T_3$ satisfies both Properties~$(i)$ and~$(ii)$.

We first prove Property~$(i)$. By Lemma~\ref{lemma:points_in_square_lemma}, $|T_2| = O(c \varepsilon^{-1})$. Then for $e \in E_r$, the length of the edge $e \cap S'$ is at least $\varepsilon r /2$. By the $c$-packedness property on $S'$, we have $c \cdot 4r \geq |E_r| \cdot \varepsilon r / 2$. Therefore, $|E_r| = O(c \varepsilon^{-1})$. The set~$T_3$ consists of $O(\varepsilon^{-1})$ points per edge in $|E_r|$, so $|T_3| = O(c \varepsilon^{-2})$. This completes the proof of Property~$(i)$. 

Next we prove Property~$(ii)$. Let $f \in F \cap S$. We have three cases, either~$f$ is a graph vertex,~$f$ is on an edge with length $\leq \varepsilon r /2$, or~$f$ is on an edge with length $\geq \varepsilon r/2$. If~$f$ is a graph vertex, then Lemma~\ref{lemma:points_in_square_lemma} implies that there exists $t \in T_2$ so that $\graphdist(f,t) \leq \varepsilon r/2$. If~$f$ is on an edge with length $\leq \varepsilon r /2$, let~$v$ be one of the endpoints of the edge. There exists $t \in T_2$ so that $\graphdist(v,t) \leq \varepsilon r/2$. Therefore, $\graphdist(f,t) \leq \graphdist(f,v) + \graphdist(v,t) \leq \varepsilon r$. Finally, if~$f$ is on an edge with length $\geq \varepsilon r/2$, then let the edge be~$e$. There exists $O(\varepsilon^{-1})$ evenly spaced points on the chord $e \cap S'$ in $T_3$. Since the distance between consecutive points is $\leq \varepsilon r/2$, there exists a point $t_3 \in T_3$ so that $\graphdist(f,t_3) \leq \varepsilon r/2$. This completes the proof of Property~$(ii)$ and we are done.
\end{proof}

Our next step is to build a data structure analogous to Lemma~\ref{lemma:starting_and_ending_point_data_structure}, but for computing points that $a_i$ can match to. To build this data structure, we first construct a three-dimensional low-density environment.

\begin{definition}
\label{definition:low_density}
A set of objects in $\mathbb R^3$ is $k$-low-density if, for every axis-parallel cube $H_r$ with side length $r$, there are at most $k$ objects that satisfy~$(i)$ the object intersects $H_r$, and~$(ii)$ the object has size at least $r$. The size of an object is the side length of the smallest axis-parallel cube that encloses the object.
\end{definition}

In our three dimensional space, the point $(x,y,z) \in \mathbb R^3$ represents a disc~$D(x,y,z) \subset \mathbb R^2$ with centre $(x,y)$ and radius~$z$. To compute a set of points satisfying the properties of Lemma~\ref{lemma:points_in_square_lemma_version_2}, we need be able to find all edges $e$ that intersect $D(x,y,z)$ and satisfy $|e| \geq \varepsilon z/2$. This motivates Definition~\ref{definition:trough}, where for each every edge $e$, we construct a three dimensional object that captures all discs $D(x,y,z)$ where~$e$ would intersect $D(x,y,z)$ and $|e| \geq \varepsilon z /2$. 

\begin{definition}
\label{definition:trough}
Given a segment $e \subset \mathbb R^2$ and $\varepsilon \in \mathbb R^+$, we define $\trough(e,\varepsilon) \subset \mathbb R^3$ to be

\[
    \trough(e,\varepsilon) = \{(x,y,z): d((x,y),e) \leq 4 z \leq 8 \varepsilon^{-1} |e|\}
\]
\end{definition}

The three-dimensional object $\trough(e,\varepsilon) \subset \mathbb R^3$ is the union of two half-cones and a triangular prism. See Figure~\ref{fig:stage3_1}. 

\begin{figure}[ht]
    \centering
    \includegraphics{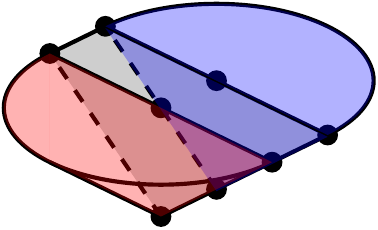}
    \caption{A trough is the union of two half-cones (red, blue) and a triangular prism (grey).}
    \label{fig:stage3_1}
\end{figure}

The base of the half-cones have a radius of $8 \varepsilon^{-1} |e|$. The half-cones and the triangular prism have a height of $2 \varepsilon^{-1} |e|$. Next, we show that if $P = (V,E)$ is a $c$-packed graph, then the set of troughs $\{\trough(e,\varepsilon): e \in E\}$ is indeed a low-density environment.

\begin{lemma}
\label{lemma:troughs_are_low_density}
Let $P = (V,E)$ be a $c$-packed graph, and let $0 < \varepsilon < 1$. Then $\{\trough(e,\varepsilon): e \in E\}$ is an $O(c\varepsilon^{-1})$-low-density environment. 
\end{lemma}

\begin{proof}
First, we show that $\trough(e,\varepsilon)$ has size at most $18 \varepsilon^{-1} |e|$. Let $(x,y,z) \in \trough(e,\varepsilon)$. Then $0 \leq z \leq 2 \varepsilon^{-1} |e|$. Moreover, $d((x,y),e) \leq 8 \varepsilon^{-1} |e|$, so $(x,y)$ must lie inside a circle centred at the midpoint of $e$, with radius $9 \varepsilon^{-1} |e|$. So~$(x,y,z)$ lies in a cylinder with radius $9 \varepsilon^{-1} |e|$ and height $2 \varepsilon^{-1} |e|$. So the size of $\trough(e,\varepsilon)$ is at most $18 \varepsilon^{-1} |e|$, as claimed.

Let $H_r$ be any axis parallel cube with side length $r$. Let $z_{min}$ be the minimum $z$-coordinate of~$H_r$. We can assume without loss of generality that $z_{min} \geq 0$. Suppose $\trough(e,\varepsilon)$ intersects with~$H_r$ and $\trough(e,\varepsilon)$ has size at least $r$. Let $(x,y,z) \in \trough(e,\varepsilon) \cap H_r$. Let $h$ be the projection of the centre of~$H_r$ onto the hyperplane defined by $z=0$. See Figure~\ref{fig:stage3_2}.

\begin{figure}[ht]
    \centering
    \includegraphics{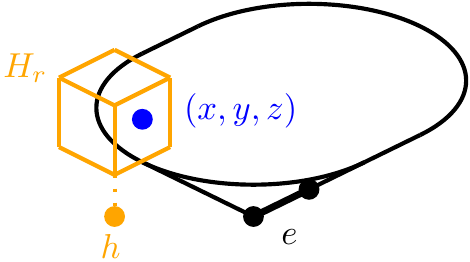}
    \caption{The point~$(x,y,z)$ (blue) is in the intersection of $\trough(e,\varepsilon)$ (black) and the cube~$H_r$ (orange). The point $h$ (orange) and the edge $e$ (black) are on the hyperplane $z=0$.}
    \label{fig:stage3_2}
\end{figure}

Then,
\[
    \begin{array}{rcl}
        d(h,e) 
        &\leq& d(h,(x,y)) + d((x,y),e) \\
        &\leq& r + 4z \\
        &\leq& 5r + 4z_{min}
    \end{array}
\]
where the first inequality is the triangle inequality, the second comes from $(x,y,z) \in H_r \cap \trough(e,\varepsilon)$, and the third comes from the maximum $z$-coordinate of $H_r$ being $z_{min} + r$.

The maximum $z$-coordinate of $\trough(e,\varepsilon)$ is $2 \varepsilon^{-1} |e|$, whereas the minimum $z$-coordinate of $H_r$ is~$z_{min}$. Therefore, $z_{min} \leq 2 \varepsilon^{-1} |e|$. Since $\trough(e,\varepsilon)$ has size at least $r$ and at most $18 \varepsilon^{-1} |e|$, we have $r \leq 18 \varepsilon^{-1} |e|$. Putting this together, we have $5r + 4z_{min} \leq 100 \varepsilon^{-1} |e|$. 

Consider the ball $B(h,10r + 8z_{min})$ centred at $h$ with radius $10r + 8z_{min}$. Since $d(h,e) \leq 5r + 4z_{min}$, the length of $e$ that is contained in $B(h,10r + 8z_{min})$ is at least $\min(|e|,5r + 4z_{min})$. But $100 \varepsilon^{-1} |e| \geq 5r + 4z_{min}$. So the length of $e$ that is contained in $B(s,10r + 8z_{min})$ is at least $\varepsilon \cdot (5r + 4z_{min})/100$. 

Finally, suppose there were $k$ edges $\{e_i\}_{i=1}^k$ so that each $e_i$ satisfies~$(i)$ $\trough(e_i)$ intersects $H_r$, and~$(ii)$ $\trough(e_i)$ has size at least $r$. Then by definition, the environment is $k$-low-density. It suffices to upper bound $k$. The total length of edges inside $B(h,10r + 8z_{min})$ is at least $k \varepsilon \cdot (5r + 4z_{min})/100$. By the $c$-packedness of~$P$, we have $c \cdot (10r + 8z_{min}) \geq k \varepsilon \cdot (5r + 4z_{min})/100$, so $k \leq 50c\varepsilon^{-1}$. Therefore, for all $H_r$ there are at most $50 c\varepsilon^{-1}$ troughs that intersect $H_r$ and have size at least $r$. This proves the stated lemma.
\end{proof}

Next, we use the result of Schwarzkopf and Vleugels~\cite{DBLP:journals/ipl/SchwarzkopfV96} to build a range searching data structure for the low-density environment. Note that $\trough(e,\varepsilon)$ has constant description complexity.

\begin{lemma}[Theorem~3 in~\cite{DBLP:journals/ipl/SchwarzkopfV96}]
\label{lemma:schwarzkopf_and_vleugels}
Let $\mathcal E$ be a set of $n$ objects in $\mathbb R^3$, where each object has constant description complexity. Suppose that $\mathcal E$ is a $k$-low-density environment. Then $\mathcal E$ can be stored in a data structure of size $O(n \log^2 n + kn)$, such that it takes $O(\log^2 n + k)$ time to report all objects that contain a given query point $q \in \mathbb R^3$. The data structure can be computed in $O(n \log^2 n + kn \log n)$ time.
\end{lemma}

We use Lemma~\ref{lemma:schwarzkopf_and_vleugels} to construct a data structure for computing points that the vertices of the query trajectory (i.e. $a_i$ for $1 \leq i \leq q$) can match to. In particular, for any square $S$ in the plane, the following data structure returns a set of points $T \subset F$ satisfying the properties in Lemma~\ref{lemma:points_in_square_lemma_version_2}.

\begin{lemma}
\label{lemma:starting_and_ending_points_data_structure_version_2}
Let $P = (V,E)$, let $F = \{f \in \mathbb R^2: f \in e, e \in E\}$, and let $\varepsilon > 0$. One can construct a data structure in $O(p^2 \log p)$ time of size $O(p \log^2 p)$, so that given a query square $S$ in the plane with side length $2r$, the data structure returns in $O(\log^2 p + c \varepsilon^{-2})$ time a set of points $T \subset F$ satisfying~$(i)$ $|T| = O(c \varepsilon^{-2})$ and~$(ii)$ for all points $f \in F \cap S$, there exists $t \in T$ so that $\graphdist(f,t) \leq \varepsilon r$.
\end{lemma}

\begin{proof}
We construct the data structure in Lemma~\ref{lemma:starting_and_ending_point_data_structure}. For each edge $e \in E$, we construct $\trough(e,\varepsilon)$, and we use Lemma~\ref{lemma:schwarzkopf_and_vleugels} to construct a range searching data structure on the set of troughs. Note that troughs form a low-density environment by Lemma~\ref{lemma:troughs_are_low_density}. This completes the construction procedure.

Given a query square $S$, we use Lemma~\ref{lemma:starting_and_ending_point_data_structure} to query a set $T_2$. Next, we state the query for $T_3$. Let the centre of $S$ be $(x,y)$, and its side length be $2r$. Query the data structure in Lemma~\ref{lemma:schwarzkopf_and_vleugels} for all troughs that contain the query point $(x,y,r)$. Suppose the data structure returns $\{\trough(e_i)\}_{i=1}^k$. Let $S'$ be the square concentric with $S$, but with side length $4r$. For each $e_i$, choose $O(\varepsilon^{-1})$ evenly spaced points on the chord $e_i \cap S'$, so that the distance between consecutive points is $\leq \varepsilon r / 2$. This completes the query for set~$T_3$.

Next, we prove the correctness of the query. For $T_2$, the proof of correctness follows from Lemma~\ref{lemma:starting_and_ending_point_data_structure}. For $T_3$, we require all edges with length at least $\varepsilon r / 2$ that intersect $S'$. It suffices to show that querying Lemma~\ref{lemma:schwarzkopf_and_vleugels} for all troughs containing the query point $(x,y,r)$ is sufficient to obtain all such edges. Recall from the definition of the trough that $(x,y,r) \in \trough(e,\varepsilon)$ if and only if $d((x,y),e) \leq 4r$ and $4r \leq 8 \varepsilon^{-1}|e|$. Note that $d((x,y),e) \leq 4r$ covers all edges that intersect $S'$, and $4r \leq 8 \varepsilon^{-1}|e|$ covers all edges with length at least $\varepsilon r/2$. Therefore, the queries for $T_2$ and $T_3$ are correct. Lemma~\ref{lemma:points_in_square_lemma_version_2} proves that $T_2 \cup T_3$ satisfies Properties~$(i)$ and~$(ii)$ in the stated lemma.

The data structure in Lemma~\ref{lemma:starting_and_ending_point_data_structure} has $O(p \log p)$ size. The data structure in Lemma~\ref{lemma:schwarzkopf_and_vleugels} has $O(p \log^2 p + c p) = O(p \log^2 p)$ size. Therefore, the overall size of our data structure is $O(p \log^2 p)$. 

The preprocessing time of Lemma~\ref{lemma:starting_and_ending_point_data_structure} and Lemma~\ref{lemma:schwarzkopf_and_vleugels} are $O(p^2 \log p)$ and $O(p \log^2 p + c \varepsilon^{-1} p \log p)$ respectively. Therefore, the overall preprocessing time is $O(p^2 \log p)$.

Finally, the query time of Lemma~\ref{lemma:starting_and_ending_point_data_structure} is $O(\log p  + c \varepsilon^{-1})$. The query time of Lemma~\ref{lemma:schwarzkopf_and_vleugels} is $O(\log^2 n + c\varepsilon^{-1})$ time. Constructing the evenly spaced points takes $O(\varepsilon^{-1})$ time per chord, and there are $O(c\varepsilon^{-1})$ chords overall. Therefore, the overall query time is $O(\log^2 p + c\varepsilon^{-2})$. 
\end{proof}

Finally, we are ready to construct the map matching data structure for general trajectory queries on $c$-packed graphs.

\main*

\begin{proof}
The preprocessing procedure is to build the data structures in Lemmas~\ref{lemma:map_matching_segment_query_if_u_and_v_are_known} and~\ref{lemma:starting_and_ending_points_data_structure_version_2}. For both data structures, we use the parameter $\varepsilon' = \varepsilon/9$ instead of~$\varepsilon$.

Given a query trajectory $Q$ and a Fr\'echet distance of $r$, the decision problem is to decide whether~$r^* \leq r$ or~$r^* \geq r$, where $r^* = \min_{\pi} \frechet(\pi, Q)$ as~$\pi$ ranges over all paths in~$P$ that start and end at a vertex of~$P$. Recall that the vertices of~$Q$ are $a_1,\ldots,a_q$. We divide the decision problem into three steps. In step one, we construct a square of side length $2r$ centred at~$a_i$. For $1 \leq i \leq q$, we query Lemma~\ref{lemma:starting_and_ending_points_data_structure_version_2} with this square to obtain a set of points $T_i$ that can match to~$a_i$. In step two, we build a directed graph over~$\cup_{i=1}^q T_i$, which we define as follows. Let~$b_{i,j} \in T_i$ and~$b_{i+1,k} \in T_{i+1}$. Let~$c_{i,j}$ and~$d_{i,j}$ be graph vertices so that~$b_{i,j}$ is on the edge~$c_{i,j}d_{i,j}$. Define~$c_{i+1,k}$ and~$d_{i+1,k}$ analogously. If $b_{i,j}$ and $b_{i+1,k}$ lie on the same edge, then the Fr\'echet distance $d_F(b_{i,j}b_{i+1,k},a_ia_{i+1})$ is the smaller of the two lengths $|b_{i,j}a_i|$ and $|b_{i+1,k}a_{i+1}|$. Otherwise, $b_{i,j}$ and $b_{i+1,k}$ lie on two different edges, and we can suppose without loss of generality that the path~$\pi$ passes through the pair of endpoints~$(c_{i,j},c_{i+1,k})$. Analogous arguments can be made if the path~$\pi$ instead passes through the pairs of endpoints $(c_{i,j},d_{i+1,k})$, $(d_{i,j},c_{i+1,k})$ or $(d_{i,j},d_{i+1,k})$. Let~$a_i'$ be the point on~$a_ia_{i+1}$ that is the closest to~$a_i$ and satisfies $\frechet(b_{i,j}c_{i,j}, a_ia_i') \leq r$. Let~$a_{i+1}'$ be the point on~$a_ia_{i+1}$ that is the closest to~$a_{i+1}$ and satisfies $\frechet(c_{i+1,k}b_{i+1,k}, a_{i+1}'a_{i+1}) \leq r$. See Figure~\ref{fig:stage3_3}.

\begin{figure}[ht]
    \centering
    \includegraphics{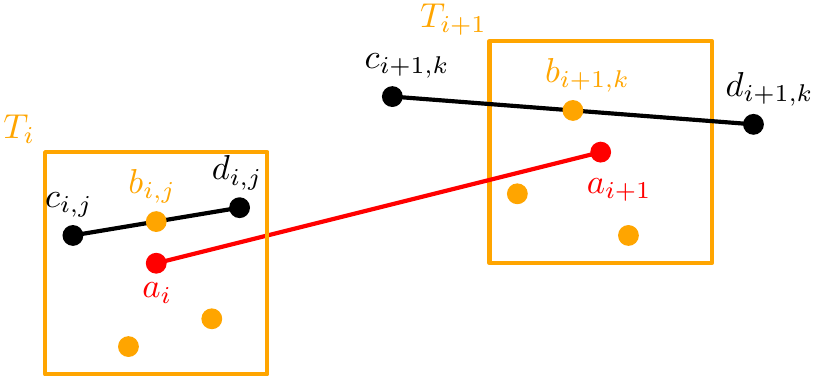}
    \caption{If~$b_{i,j} \in T_i$ (orange) is a point that can match to $a_i$ (red), we let the edge containing~$b_{i,j}$ be $c_{i,j}d_{i,j}$ (black). The analogous edge $c_{i+1,k}d_{i+1,k}$ is defined for~$b_{i+1,k}$.}
    \label{fig:stage3_3}
\end{figure}

We query the data structure in Lemma~\ref{lemma:map_matching_segment_query_if_u_and_v_are_known} to obtain a $(1+\varepsilon')$-approximation of $\min_{\pi} \frechet(\pi, a_i'a_{i+1}')$ where $\pi$ ranges over all paths between the graph vertices~$c_{i,j}$ and~$c_{i+1,k}$. In our directed graph over~$\cup_{i=1}^q T_i$, define the capacity of the directed edge from~$b_{i,j}$ to~$b_{i+1,k}$ as follows. If $b_{i,j}$ and $b_{i+1,k}$ are on the same edge, then the capacity is $\min(|b_{i,j}a_i|,|b_{i+1,k}a_{i+1}|)$. If $b_{i,j}$ and $b_{i+1,k}$ are on different edges, the capacity is the minimum of the four $(1+\varepsilon')$-approximations of $\min_{\pi} \frechet(\pi, a_i'a_{i+1}')$ where $\pi$ ranges over all paths between the pairs of vertices~$(c_{i,j},c_{i+1,k})$, $(c_{i,j},d_{i+1,k})$, $(d_{i,j},c_{i+1,k})$ and $(d_{i,j},d_{i+1,k})$. This completes step two, that is, building the directed graph. The third step is, for $r' \in \{r, (1+\varepsilon')r\}$, to decide whether there exists a path in the directed graph from $T_1$ to $T_q$, so that the capacity of each edge in the path is at most~$r'$. We distinguish three cases (a), (b) and (c): 

\begin{enumerate}[label=(\alph*), noitemsep]
    \item If there exists a path in the case $r' = r$, we return that $r^* \leq r$. 
    \item If there does not exist a path in the case $r' = (1+\varepsilon')^2 r$, we return that $r^* \geq r$.
    \item If there exists a path in the case $r' = (1+\varepsilon')^2 r$ but not for the case $r' = r$, we return that $$r^* \in [(1+\varepsilon')^{-2} r, (1+\varepsilon')^2 r].$$
\end{enumerate}

Note that the third case does not technically answer the decision problem, however, in this case, we will show that $(1+\varepsilon')^2 r$ is a $(1+\varepsilon)$-approximation of $r^*$, as required by the theorem statement. This completes the description of the query procedure in the decision version. Next, we prove its correctness, which we separate into cases (a), (b) and (c).

\begin{enumerate}[label=(\alph*), noitemsep]
    \item Suppose there exists a path with capacity at most $r'$, where $r' \geq r$. Let this path be~$b_1,\ldots,b_q$, where~$b_i \in T_i$ for $1 \leq i \leq q$. Let the capacity of the directed edge from~$b_i$ to~$b_{i+1}$ be~$C_i$. Then~$C_i \leq r'$, by definition. If $b_i$ and $b_{i+1}$ are on the same edge, define $\pi'_i$ to be the edge $b_ib_{i+1}$, so that $d_F(\pi'_i,a_ia_{i+1}) \leq r'$. Otherwise, there exists graph vertices~$c_i$ and $c_{i+1}$, and points $a_i'$ and $a_{i+1}'$ on $a_i a_{i+1}$ satisfying $\frechet(b_{i}c_{i}, a_ia_i') \leq r$, $\frechet(b_{i+1}c_{i+1}, a_{i+1}'a_{i+1}) \leq r$, and $\frechet(\pi_i, a_i'a_{i+1}') \leq C_i$ for some path $\pi_i$ between $c_{i}$ and $c_{i+1}$. Define $\pi_i'$ to be the concatenation of~$b_ic_i$, $\pi_i$ and $c_{i+1}b_{i+1}$. So $\frechet(\pi_i', a_ia_{i+1}) \leq \max(C_i,r) \leq r'$. Define $\pi'$ to be the concatenation of $\pi_i'$ for all $1 \leq i \leq q$. Then $\frechet(\pi', Q) \leq r'$. Therefore, $r^* = \min_{\pi} \frechet(\pi, Q) \leq \frechet(\pi',Q) \leq r'$, so $r^* \leq r'$. In the first case, there exists a path for $r' = r$, so $r^* \leq r$, as required.
    \item Suppose there does not exist a path with capacity at most $r'$. Let $\pi^*$ be the path in~$P$ so that $r^* = \frechet(\pi^*, Q)$. Let the points on~$\pi^*$ that match to $a_1,\ldots,a_q \in Q$ be $u_1^*,\ldots,u_q^* \in P$. By Lemma~\ref{lemma:starting_and_ending_points_data_structure_version_2}, there exist points~$b_{i} \in T_i$ so that $\graphdist(b_{i}, u_i^*) \leq \varepsilon' r$, for all $1 \leq i \leq q$. Let~$r_i$ be the minimum Fr\'echet distance $\frechet(\pi,a_ia_{i+1})$ where $\pi$ ranges over all paths between~$b_{i}$ and~$b_{i+1}$. Then~$r_i \leq \max(\frechet(b_{i} u_i^*, a_i), \frechet(\pi^*[u_i^*,u_{i+1}^*],a_ia_{i+1}), \frechet(u_{i+1}^* b_{i+1}, a_{i+1}))$, since the concatenation of~$b_{i}u_i^*$, the subtrajectory $\pi^*[u_i^*,,u_{i+1}^*]$ of $\pi^*$, and $u_{i+1}^* b_{i+1}$ is a valid path from~$b_i$ to~$b_{i+1}$. See Figure~\ref{fig:stage3_4}.

    Note that $\frechet(b_iu_i^*,a_i) \leq \frechet(u_i^*,a_i) + \graphdist(b_i, u_i^*) \leq r^* + \varepsilon' r$. Therefore,~$r_i \leq r^* + \varepsilon' r$. Then, the capacity of the edge from~$b_i$ to~$b_{i+1}$ is at most $(1+\varepsilon') r_i \leq (1+\varepsilon')(r^* + \varepsilon' r)$. Putting this together, there exists a path from $T_1$ to $T_q$ with capacity at most $(1+\varepsilon')(r^* + \varepsilon' r)$. In the second case, there does not exist a path with capacity $r' = (1+\varepsilon')^2 r$. Therefore, $(1+\varepsilon')^2 r \leq (1+\varepsilon')(r^* + \varepsilon' r)$ which implies $(1+\varepsilon') r \leq r^* + \varepsilon' r$ and $r \leq r^*$, as required. 

    \item From proof of the first case, $r^* \leq r'$, if there exists a path for $r'$. Therefore, $r^* \leq (1+\varepsilon')^2 r$. From the proof of the second case, $r^* \geq r'$ if there exists a path for $(1+\varepsilon')^{2} r'$. Therefore, $r^* \geq (1+\varepsilon)^{-2} r$. Putting this together, we get $r^* \in [(1+\varepsilon')^{-2} r, (1+\varepsilon')^2 r]$. In particular, we have $(1+\varepsilon)^2 r \geq r^*$, and $(1+\varepsilon)^2 r \leq (1+\varepsilon')^4 r^* \leq (1+3 \varepsilon')^2 r^* \leq (1+9 \varepsilon') r^* = (1+\varepsilon) r^*$, so $(1+\varepsilon')^2 r$ is a $(1+\varepsilon)$-approximation of $r^*$, as required.
\end{enumerate}

    \begin{figure}[ht]
        \centering
        \includegraphics{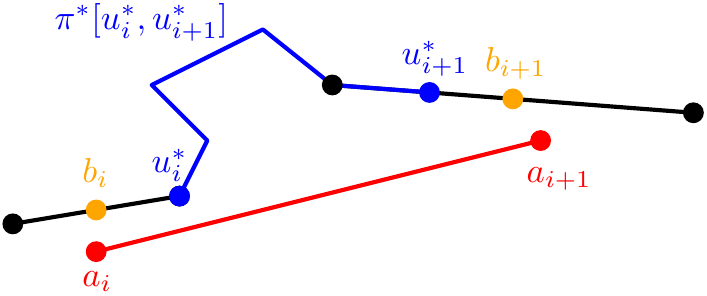}
        \caption{The trajectory vertices $a_i$ and $a_{i+1}$ (red) match to $u_i^*$ and $u_{i+1}^*$ on the optimal path $\pi^*$ (blue). We construct a path from~$b_i$ to~$b_{i+1}$ (orange) by concatenating the edge~$b_i u_i^*$, the subtrajectory $\pi^*[u_i^*,u_{i+1}^*]$ and $u_{i+1}^*b_{i+1}$.}
        \label{fig:stage3_4}
    \end{figure}

For the minimisation version, we apply parametric search. Define the decision problem $D(r)$ to be TRUE if there exists a path for $r'=r$, and FALSE if there does not exist a path for $r' = (1+\varepsilon')^2 r$. 

It suffices to show~$(i)$ that $D(r)$ is monotone and~$(ii)$ that all operations in $D(r)$ are either independent of $r$, or can be made equivalent to a constant number of comparisons $\{r > c_i\}$ where~$c_i$ is a critical value. First we show~$(i)$. Suppose $D(r_1)$ evaluates to TRUE, and~$r_1 < r_2$. Then $r^* \leq r_1 \leq r_2$ and we cannot have $D(r_2)$ evaluation to FALSE. So either $D(r_2)$ is also TRUE, or we halt the parametric search and obtain a $(1+\varepsilon)$-approximation of $r^*$. Similarly, if $D(r_1)$ evaluates to FALSE, and~$r_1 > r_2$, then we cannot have $D(r_2)$ evaluating to TRUE. Therefore, $D(r)$ is monotone. 

Next, we show~$(ii)$. The first step is to query Lemma~\ref{lemma:map_matching_segment_query_if_u_and_v_are_known} to obtain a set of points $T_i$ that can match to $a_i$. The critical values of this step are when the query point $(x,y,r)$ lies on the boundary of the troughs in the low-density environment. This can be evaluated as a low order polynomial in terms of $r$. The second step is to compute the points $a_i$ and $a_i'$, and query the data structure in Lemma~\ref{lemma:map_matching_segment_query_if_u_and_v_are_known} to obtain a $(1+\varepsilon')$-approximation of $\min_{\pi} \frechet(\pi, a_i'a_{i+1}')$. The critical values occur when $a_i$ and $a_{i}'$ match to different pairs of grid points in Lemma~\ref{lemma:exponential_grid_between_transit_pairs}, which can be resolved using a low order polynomial in terms of $r$. The third step is to decide whether there exists a path in the directed graph where the capacity of each edge is at most~$r'$. The critical values are when the capacity of an edge is exactly $r'$, which can be resolved as a low order polynomial of~$r$. This completes the proof of~$(i)$ and~$(ii)$. 

We analyse the space and preprocessing time of the data structure. First, we analyse the space. The data structures in Lemmas~\ref{lemma:map_matching_segment_query_if_u_and_v_are_known} and~\ref{lemma:starting_and_ending_points_data_structure_version_2} require $O(c \varepsilon^{-4} \log(1/\varepsilon) p \log p)$ and $O(p \log^2 p)$ space respectively. The overall space requirement is $O(c \varepsilon^{-4} \log(1/\varepsilon) p \log p + p \log^2 p)$. Next, the preprocessing times of Lemmas~\ref{lemma:map_matching_segment_query_if_u_and_v_are_known} and~\ref{lemma:starting_and_ending_points_data_structure_version_2} are $O(c^2 \varepsilon^{-4} \log^2(1/\varepsilon) p^2 \log^2 p)$ and $O(p^2 \log p)$ respectively. The overall preprocessing time is $O(c^2 \varepsilon^{-4} \log^2(1/\varepsilon) p^2 \log^2 p)$.

For the query time of the decision version, for each $1 \leq i \leq q$ we query the data structure in Lemma~\ref{lemma:points_in_square_lemma_version_2} to construct the set $T_i$. In total, this takes $O(q \cdot (\log^2 p + c \varepsilon^{-2}))$ time. Next, we build a directed graph on $\cup_{i=1}^q T_i$. There are $O(q \cdot c^2 \varepsilon^{-4})$ directed edges between $T_i$ and $T_{i+1}$, for $1 \leq i < q$. Computing the capacity of the directed edge, by querying Lemma~\ref{lemma:map_matching_segment_query_if_u_and_v_are_known}, takes $O(\log p + c \varepsilon^{-1})$ time. Finally, deciding whether there is a directed path from $T_1$ to $T_q$ takes $O(q \cdot c^2 \varepsilon^{-4})$ time. Therefore, the overall running time of the decision version is $O(q \cdot (\log^2 p + c^2 \varepsilon^{-4} \log p))$.

Finally, we analyse the running time of parametric search. The running time of parametric search~\cite{DBLP:journals/jacm/Megiddo83} is $O(N_p T_p + T_p T_s \log N_p)$, where $T_s$ is the running time of the sequential algorithm, $N_p$ is the number of processors for the parallel algorithm, and $T_p$ is the number of parallel steps in the parallel algorithm. The sequential algorithm is the same as the decision algorithm, so $T_s = O(q \cdot (\log^2 p + c^2 \varepsilon^{-4} \log p))$. The parallel algorithm is to simulate the decision algorithm on $N_p = q$ processors. The first two steps can be parallelised to run on~$q$ processors in $O(\log^2 p + c^2 \varepsilon^{-4} \log p)$ parallel steps. In the third step, it suffices to check if each of the edges has capacity at most~$r'$. Computing the directed path generates no additional critical values, so it does not need to be simulated by the parallel algorithm. The third step can be parallelised onto~$q$ processors to run in $O(c^2 \varepsilon^{-4})$ parallel steps. The total number of parallel steps is $T_p = O(\log^2 p + c^2 \varepsilon^{-4} \log p)$. Therefore, substituting these values into the running time of parametric search, we get that the overall running time of parametric search is $O(q \log q \cdot (\log^4 p + c^4 \varepsilon^{-8} \log^2 p))$. 
\end{proof}

\section{Lower bound for geometric planar graphs}
\label{sec:lower_bounds}

In this section, we no longer assume that the graph~$P$ is $c$-packed. The main result of this section is that for geometric planar graphs, unless SETH fails, no data structure can preprocess a graph in polynomial time to answer map matching queries in $O((pq)^{1-\delta})$ time for any~$\delta > 0$, and for any polynomial restrictions of~$p$ and~$q$. In Section~\ref{sec:lower_bound;subsec:warmup}, we construct a lower bound for the warm-up problem of Fr\'echet distance queries on trajectories. In Section~\ref{sec:lower_bound;subsec:map_matching}, we construct a lower bound for map matching queries on geometric planar graphs.

\subsection{Fr\'echet distance queries on trajectories}
\label{sec:lower_bound;subsec:warmup}

Our warm-up problem is to extend the lower bound of Bringmann~\cite{DBLP:conf/focs/Bringmann14} to Fr\'echet distance queries on trajectories. The lower bound assumes a weaker version of SETH. 

\begin{definition}[SETH$'$]
The CNF-SAT problem is to decide whether a formula $\varphi$ on $N$ variables $x_1, \ldots x_N$ and $M$ clauses~$C_1, \ldots, C_M$ has a satisfying assignment. SETH$\,'$ states that there is no $\tilde O((2-\delta)^N)$ algorithm for CNF-SAT for any $\delta > 0$, where $\tilde O$ hides polynomial factors in~$N$ and~$M$.
\end{definition}

If SETH$'$ fails, then so does SETH~\cite{DBLP:journals/tcs/Williams05}. Next, we provide a proof sketch of Theorem~1.2 in~\cite{DBLP:conf/focs/Bringmann14}.

\begin{lemma}[Theorem 1.2 in~\cite{DBLP:conf/focs/Bringmann14}] 
\label{lemma:bringmann}
Let $n$ and $m$ denote the complexities of a pair of trajectories. There is no $1.001$-approximation with running time $O((nm)^{1-\delta})$ for the Fr\'echet distance for any $\delta > 0$, unless SETH fails. This holds for any polynomial restrictions of $n$ and $m$.
\end{lemma}

\begin{proof}[Proof (Sketch)]
Suppose for the sake of contradiction that there exists a positive constant $\delta$ so that there is a $1.001$-approximation with running time $O((nm)^{1-\delta})$ for the Fr\'echet distance. We summarise the main steps of Theorem~1.2 in~\cite{DBLP:conf/focs/Bringmann14}, which generalises the $m=n$ case to the $m \neq n$ case. 

Suppose $m = \Theta(n^\gamma)$ for some $\gamma > 0$. We are given a CNF-SAT instance $\varphi$ with variables $x_1, \ldots, x_N$ and clauses~$C_1, \ldots, C_M$. We partition its variables $x_1, \ldots, x_N$ into $V_1 = \{x_1,\ldots,x_\ell\}$ and $V_2 = \{x_{\ell+1},\ldots,x_N\}$, where $\ell = N/(\gamma + 1)$. Let $A_k$ be all the assignments of $V_k$ for $k \in \{1,2\}$. Using the same method as the $m=n$ case in~\cite{DBLP:conf/focs/Bringmann14}, we construct curves $P_1$ and $P_2$ so that $|P_1| = \Theta(M \cdot |A_1|)$ and $|P_2| = \Theta(M \cdot |A_2|)$. Moreover, by Lemma~3.7 and Lemma~3.9 in~\cite{DBLP:conf/focs/Bringmann14}, if $A_1 \times A_2$ contains a satisfying assignment, then $\frechet(P_1,P_2) \leq 1$, whereas if $A_1 \times A_2$ contains no satisfying assignment, then $\frechet(P_1,P_2) \geq 1.001$. 
Note that if $n = |P_1|$, then $m = |P_2| = \Theta(n^\gamma)$.

Therefore, any $1.001$-approximation of $\frechet(P_1,P_2)$ with running time $O((nm)^{1-\delta})$ yields an algorithm for CNF-SAT with running time $O((M \cdot  |A_1|)^{1-\delta} (M \cdot |A_2|)^{1-\delta}) = O(M^2 2^{(1-\delta)N})$, contradicting SETH$'$ and SETH.
\end{proof}

Next, we consider the problem of preprocessing a trajectory such that, given a query trajectory, their Fr\'echet distance can be computed efficiently. This is stated as an extremely challenging problem in Buchin~et~al.~\cite{DBLP:conf/esa/BuchinHOSSS22}. In Lemma~\ref{lemma:frechet_distance_query_lower_bound} we show that preprocessing essentially does not help with computing the Fr\'echet distance. In particular, we show that even with polynomial preprocessing time, one cannot obtain a truly subquadratic query time for computing the Fr\'echet distance. We prove this by considering the offline version of the data structure problem, in a similar fashion to Rubinstein~\cite{DBLP:conf/stoc/Rubinstein18}, Driemel and Psarros~\cite{DBLP:journals/corr/driemelpsarros} and Bringmann~et~al.~\cite{DBLP:conf/soda/BringmannDNP22}. 

\begin{lemma}
\label{lemma:frechet_distance_query_lower_bound}
Let $n$ denote the complexity of a trajectory. There is no data structure that can be constructed in $\poly(n)$ time, that when given a query trajectory of complexity $m$, can answer $1.001$-approximate Fr\'echet distance queries in $O((nm)^{1-\delta})$ query time for any $\delta > 0$, unless SETH fails. This holds for any polynomial restrictions of $n$ and $m$.
\end{lemma}

\begin{proof}
Suppose for the sake of contradiction that there exists positive constants $\alpha$ and $\delta$ so that one can construct a data structure in $O(n^\alpha)$ preprocessing time to answer $1.001$-approximate Fr\'echet distance queries with a query time of $O((nm)^{1-\delta})$. 

Suppose $m = \Theta(n^\gamma)$ for some $\gamma > 0$. We take two cases. In the first case, $\gamma \geq 2 \alpha$. Given a pair of trajectories with complexities $n$ and $m$, we can preprocess the first trajectory in $O(n^\alpha)$ time, and query a $1.001$-approximation of its Fr\'echet distance with the second trajectory in $O((nm)^{1-\delta})$ time. But $O(n^\alpha) = O(m^{1/2})$, so the overall running time is $O((nm)^{1-\delta})$. This contradicts Lemma~\ref{lemma:bringmann}.

In the second case, $\gamma \leq 2 \alpha$. We follow the same steps as Lemma~\ref{lemma:bringmann}. We are given a CNF-SAT instance $\varphi$ with variables $x_1,\ldots,x_N$ and clauses~$C_1,\ldots,C_M$. We partition its variables $x_1,\ldots,x_N$ into $V_1 = \{x_1,\ldots,x_\ell\}$ and $V_2 = \{x_{\ell+1},\ldots,x_N\}$, where $\ell = N /(2\alpha + 1)$. Let $A_k$ be all the assignments of $V_k$ for $k \in \{1,2\}$. Note that if $n = |A_1|$, then $m = |A_2| = \Theta(n^{2\alpha})$. Partition the set~$A_2$ into subsets $B_1, \ldots, B_K$ so that $|B_i| = \Theta(|A_1|^\gamma)$, for all $1 \leq i \leq K$, and $K = O(|A_1|^{2 \alpha - \gamma})$. Note that $A_1 \times A_2$ contains a satisfying assignment if and only if there exists $1 \leq i \leq K$ so that $A_1 \times B_i$ contains a satisfying assignment. Using the same method as the $m=n$ case in~\cite{DBLP:conf/focs/Bringmann14}, we construct curves $P_1$ and $Q_i$ so that $|P_1| = \Theta(M |A_1|)$ and $|Q_i| = \Theta(M |B_i|)$ for $1 \leq i \leq K$. Moreover, by Lemma~3.7 and Lemma~3.9 in~\cite{DBLP:conf/focs/Bringmann14}, if $A_1 \times B_i$ contains a satisfying assignment, then $\frechet(P_1,Q_i) \leq 1$, whereas if $A_1 \times B_i$ contains no satisfying assignment, then $\frechet(P_1,Q_i) \geq 1.001$. Note that if $n = |P_1|$, then $m = |Q_i| = \Theta(n^\gamma)$. 

Therefore, to decide if there is a satisfying assignment for the CNF-SAT instance~$\varphi$, it suffices to query a $1.001$-approximation of $\frechet(P_1,Q_i)$ for all $1 \leq i \leq K$. We preprocess the trajectory $P_1$ in $O((M |A_1|)^\alpha) = O(M^\alpha 2^{\alpha N/(2\alpha+1)}) = O(M^\alpha 2^{N/2})$ time. We answer all $K$ queries in time
\[
    \begin{array}{rcl}
        O(\sum_{i=1}^K(M|A_1|)^{1-\delta}(M|B_i|)^{1-\delta}) 
            &=& O(KM^2|A_1|^{1-\delta} |A_1|^{(1-\delta)\gamma}) \\
            &=& O(M^2 |A_1|^{(1-\delta) + (1-\delta)\gamma + 2 \alpha - \gamma}) \\
            &=& O(M^2 |A_1|^{(1-\delta) + 2 \alpha}) \\
            &=& O(M^2 2^{N(1 + 2\alpha - \delta)/(1+2\alpha)}) \\
            &=& O(M^2 2^{(1 -\delta/(1+2\alpha))N}).
    \end{array}
\]
Putting this together, we yield an algorithm for CNF-SAT with running time 
\[
    O(M^\alpha 2^{N/2} + M^2 2^{(1 - \frac \delta {1+2\alpha})N}),
\]
where $\alpha$ and $\delta$ are constants, contradicting SETH$'$and SETH.
\end{proof}

An open problem is whether one can adapt the lower bound of Buchin~et~al.~\cite{DBLP:conf/soda/BuchinOS19} to rule out approximation factors between $1.001$ and $3$. In particular, one would need to extend their construction for pairs of trajectories with an imbalanced number of vertices to rule out approximation algorithms with running time $O((nm)^{1-\delta})$ for all $\delta > 0$. 

Another open problem is whether one can extend the lower bound to range searching queries. Given a database of $k$ trajectories with $m$ vertices each and a query trajectory with $n$ vertices, Baldus and Bringmann~\cite{DBLP:conf/gis/BaldusB17} conjecture that a $O((kmn)^{1-\delta})$ time algorithm for range searching would falsify SETH.  

\subsection{Map matching queries on geometric planar graphs}
\label{sec:lower_bound;subsec:map_matching}

We return to the map matching problem. The main result of this section is that there is no data structure that can preprocess a geometric planar graph in polynomial time to answer map matching queries in truly subquadratic time. To build towards this result, we first show that Alt~et~al.'s~\cite{DBLP:journals/jal/AltERW03} $O(pq \log p)$ time algorithm for Problem~\ref{problem:map_matching} is optimal up to lower-order factors, conditioned on unbalanced OVH.

\begin{definition}[OVH]
The OV problem is to decide whether the sets $A, B \subseteq \{0,1\}^d$ contain a pair of binary vectors $(a,b) \in A \times B$ so that~$a$ and~$b$ are orthogonal. Let $n = |A|$ and $m = |B|$. OVH states that there is no $\tilde O((nm)^{1-\delta})$ time algorithm for OV for any $\delta > 0$, where $\tilde O$ hides polynomial factors in~$d$. This holds for any polynomial restrictions of $n$ and $m$.
\end{definition}

If OVH fails, then so does SETH~\cite{DBLP:journals/tcs/Williams05}. We use OVH to prove our lower bound for Problem~\ref{problem:map_matching}. For constructing our graph and our trajectory, we use the notation $Q = \bigcirc_{i=1}^q a_i = a_1 \circ \ldots \circ a_q$ to denote the polygonal curve~$Q$ obtained by linearly interpolating between the vertices $a_1, \ldots, a_q$. 

\begin{lemma}
\label{lemma:map_matching_lower_bound}
Let~$P$ be a geometric planar graph of complexity~$p$ and~$Q$ be a trajectory of complexity~$q$. There is no $2.999$-approximation with running time $O((pq)^{1-\delta})$ for computing $\min_{\pi} \frechet(\pi,Q)$ for any $\delta > 0$, unless SETH fails. This holds for any polynomial restrictions of~$p$ and~$q$. 
\end{lemma}

\begin{proof}
Suppose for the sake of contradiction that there exists a positive constant $\delta$ so that there is a $2.999$-approximation with running time $O((pq)^{1-\delta})$ for Problem~\ref{problem:map_matching}.

We are given an OV instance $A, B \subseteq \{0,1\}^d$. Let $p = |A|$ and $q = |B|$, where there may be any polynomial restriction of~$p$ and~$q$. First, we will construct a graph~$P$ of complexity $O(dp)$ and a trajectory~$Q$ of complexity $O(dq)$. Then we will show that a $2.999$-approximation of $\min_{\pi} \frechet(\pi, Q)$ yields an $O((pq)^{1-\delta})$ time algorithm for OV.

Let $h$ be a small constant, which we will choose later on in the proof. Inspired by Buchin~et~al.~\cite{DBLP:conf/soda/BuchinOS19} and Bringmann~et~al.~\cite{DBLP:conf/soda/BringmannDNP22}, we define the following polygonal curves. See Figure~\ref{fig:lower_bound_1}. It is straightforward to verify that $\frechet(0_A,0_B) = \frechet(0_A,1_B) = \frechet(1_A,0_B) = 1$, and $\frechet(1_A,1_B) = 3$. 

\[
    \begin{array}{rcl}
        1_A &:=& (0,0) \circ (12,0) \circ (12,h) \circ (6,h) \circ (6, 2h) \circ (18,2h) \\
        0_B &:=& (0,0) \circ (13,0) \circ (13,h) \circ (5,h) \circ (5, 2h) \circ (18,2h) \\
        0_A &:=& (0,0) \circ (14,0) \circ (14,h) \circ (4,h) \circ (4, 2h) \circ (18,2h) \\
        1_B &:=& (0,0) \circ (15,0) \circ (15,h) \circ (3,h) \circ (3, 2h) \circ (18,2h)
    \end{array}
\]

\begin{figure}[ht]
    \centering
    \includegraphics{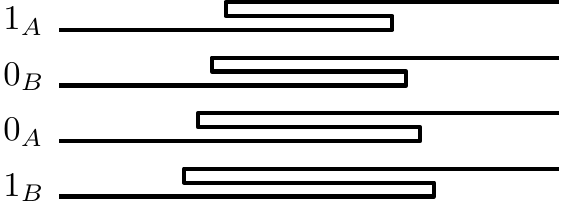}
    \caption{The polygonal curves $1_A$, $0_B$, $0_A$ and $1_B$.}
    \label{fig:lower_bound_1}
\end{figure}

We use $0_A$ and $1_A$ to construct~$P$. We start by constructing the curves $R$, $S$ and $T_i$.
\[
    \def\arraystretch{1.1}
    \begin{array}{rcl}
        R &:=& \bigcirc_{k=1}^d R_k \quad\text{where}\quad R_k := 0_A + (18k,0),\\
        S &:=& \bigcirc_{k=1}^d S_k \quad\text{where}\quad S_k := 0_A + (18k,7ph),\\
        T_i &:=& \bigcirc_{j=1}^d T_{i,k} \quad\text{where}\quad T_{i,k} := A[i][k]_A + (18k,3ih) \quad\text{for all}\quad 1 \leq i \leq p, \quad 1 \leq k \leq d\\
    \end{array}
\]
where $A[i][k]$ is the $k^{th}$ coordinate of the $i^{th}$ vector in~$A$, and $A[i][k]_A$ is either $0_A$ or $1_A$ depending on whether $A[i][k]$ is $0$ or $1$, and $+(x,y)$ translates the curve horizontally by~$x$ and vertically by~$y$. See Figure~\ref{fig:lower_bound_2}.

\begin{figure}[ht]
    \centering
    \includegraphics{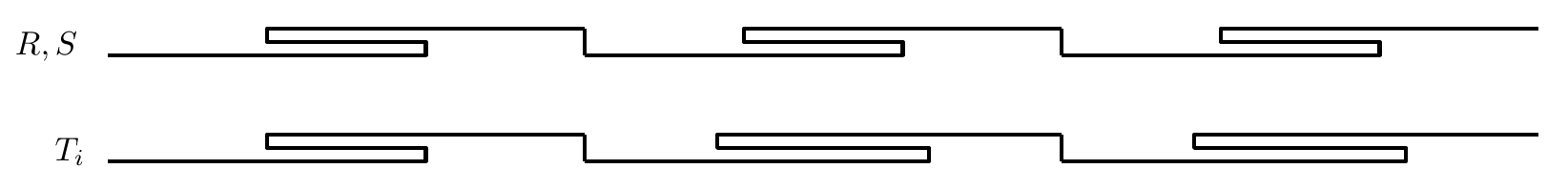}
    \caption{The curves~$R$ and~$S$ are obtained by concatenating translated versions of~$0_A$. The curves~$T_i$ are obtained by concatenating translated versions of~$0_A$ and~$1_A$.}
    \label{fig:lower_bound_2}
\end{figure}

Next, we use $R$ and $S$ to construct curves~$U$ and~$V$. Note that~$U$ and~$V$ each contain a loop.
\[
    \begin{array}{rcl}
        U &:=& (0,-18) \circ (0,0) \circ R \circ (36d,3h) \circ (0,3h) \circ (0,0), \\
        V &:=& (0,6ph) \circ S \circ (36d,6ph) \circ (0,6ph) \circ (0,18),
    \end{array}
\]

See Figure~\ref{fig:lower_bound_3}, left. Finally, we connect the curves $T_i$,~$U$ and~$V$ to obtain the graph~$P$. For $1 \leq i \leq p$, we connect~$U$ to~$T_i$ with the edge $(0,3h) \circ (18,3ih)$. For $1 \leq i \leq p$, we connect~$T_i$ to~$V$ with the edge $(18d,3ih+2h) \circ (36d, 6ph)$. We take the union of $T_i$,~$U$,~$V$, and these $2p$ connections to obtain the graph~$P$, completing its construction. See Figure~\ref{fig:lower_bound_3}, right. It is straightforward to verify that~$P$ is connected and planar, and $|P| = O(dp)$. 

\begin{figure}[ht]
    \centering
    \includegraphics[width=\textwidth]{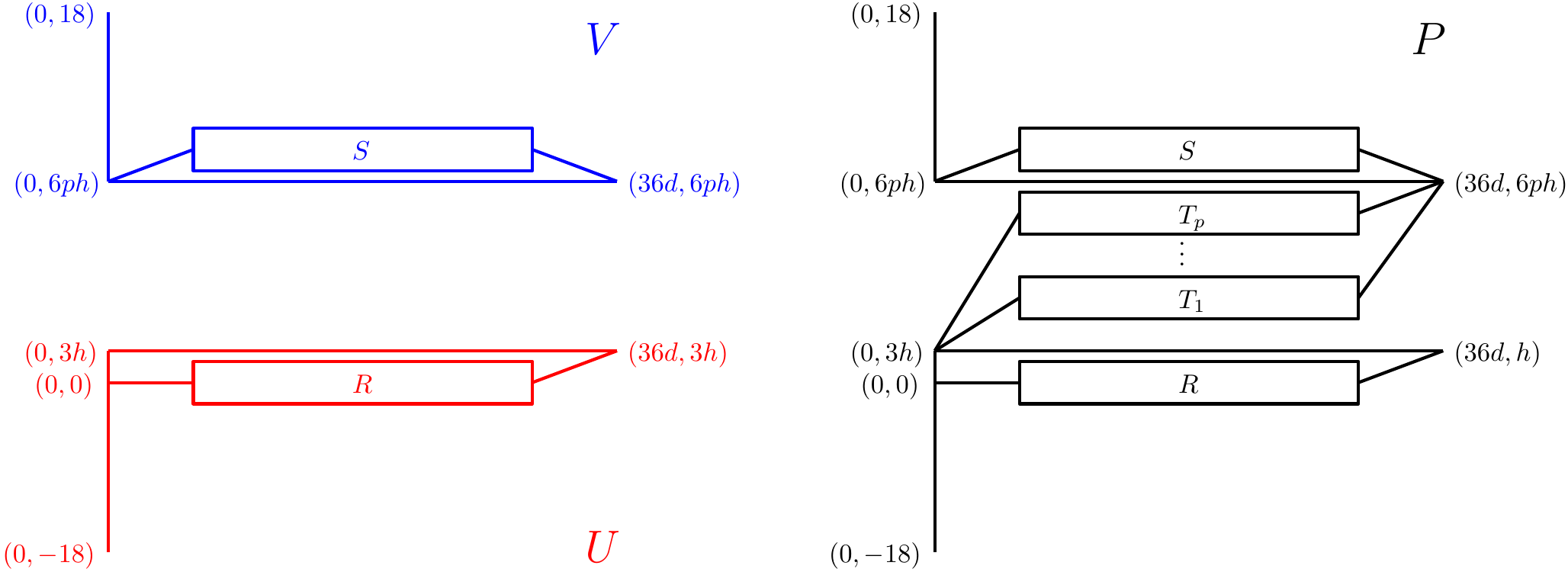}
    \caption{(Left) The lower curve~$U$ in red, the upper curve~$V$ in blue. (Right) The graph~$P$ obtained by connecting~$U$,~$V$ and $T_i$.}
    \label{fig:lower_bound_3}
\end{figure}

Now, we use $0_B$ and $1_B$ to construct~$Q$.
\[
    \begin{array}{rcl}
        W_j &:=& \bigcirc_{j=1}^d W_{j,k} \quad\text{where}\quad W_{j,k} := B[j][k]_B + (18k,0) \quad\text{for all}\quad 1 \leq j \leq q,\quad 1\leq k \leq d,\\
        X &:=& \bigcirc_{j=1}^q X_j \quad\text{where}\quad X_j := (0,0) \circ W_j \circ (36d,0) \circ (0,3h) \quad\text{for all}\quad 1 \leq j \leq q,\\
        Q &:=&  (0,-18) \circ X \circ (0,18).
    \end{array}
\]

Note that $B[j][k]$ is the $k^{th}$ coordinate of the $j^{th}$ vector in~$B$, where $B[j][k]_B$ is either $0_B$ or $1_B$ depending on whether $B[j][k]$ is $0$ or $1$, and $+(x,y)$ translates the curve horizontally by~$x$ and vertically by~$y$. This completes the construction of~$Q$. See Figure~\ref{fig:lower_bound_4}. It is straightforward to verify that $|Q| = O(dq)$. 

\begin{figure}[ht]
    \centering
    \includegraphics[width=\textwidth]{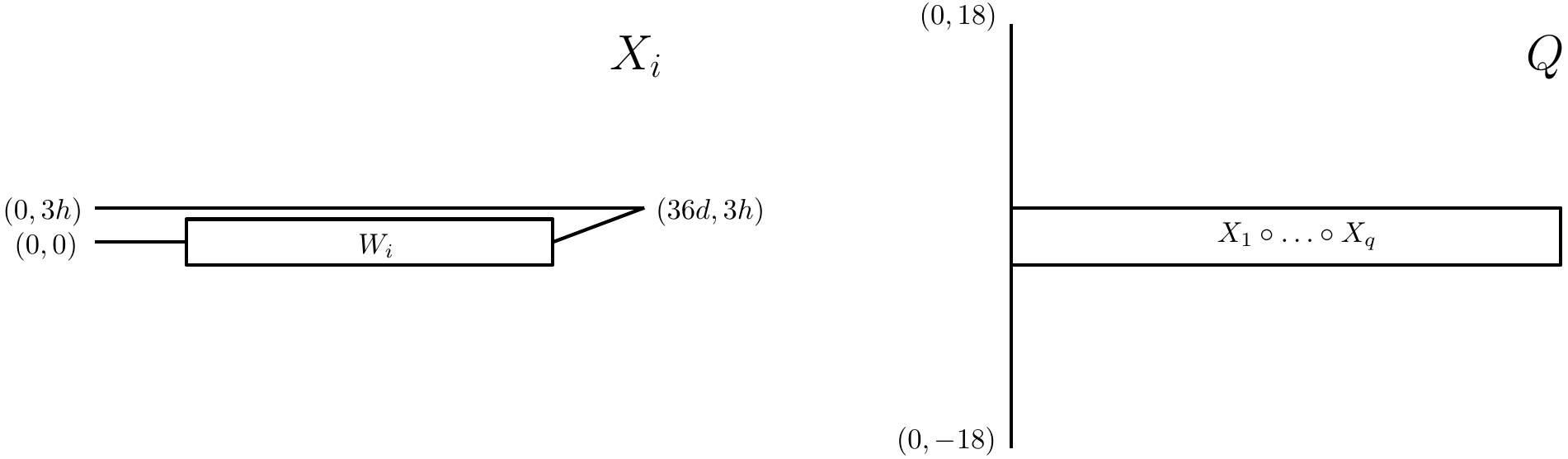}
    \caption{(Left) Curve $X_i$ is obtained by concatenating $(0,0)$, $W_i$, $(36d,3h)$ and $(0,3h)$. (Right) Curve~$Q$ is obtained by concatenating $(0,-18)$, $X_i$ for $1 \leq i \leq q$, and $(0,18)$.}
    \label{fig:lower_bound_4}
\end{figure}

We will show that if our OV instance $A,B$ is a YES-instance, then $\min_{\pi} \frechet(\pi,Q) \leq 1.001$. Suppose that $A[i]$ and $B[j]$ are orthogonal. We will construct a path $\pi \in P$ with $\frechet(\pi,Q) \leq 1.001$, which we define as follows.
\[
    \begin{array}{rclcl}\
        Y &:=& \bigcirc_{\ell=1}^{j-1} Y_\ell \quad\text{where}\quad Y_\ell := (0,0)  \circ R \circ (36d,3h) \circ (0,3h) \quad\text{for all}\quad 1 \leq \ell \leq j-1,\\
        Z &:=& \bigcirc_{\ell=j+1}^q Z_\ell \quad\text{where}\quad Z_\ell := (0,6ph) \circ S \circ (36d,6ph) \circ (0,6ph) \quad\text{for all}\quad j+1 \leq \ell \leq q,\\
        \pi &:=& (0,-18) \circ Y \circ T_i \circ Z \circ (0,18),
    \end{array}
\]

It is straightforward to verify that $\pi$ is a path of~$P$. Next, we provide a matching of~$\pi$ and~$Q$ with Fr\'echet distance at most $1.001$, assuming $h$ is sufficiently small.
\[
    \begin{array}{rclcrcl}
        \pi &=& (0,-18) \circ (0,0)
            &\qquad& 
                Q &=& (0,-18) \circ (0,0) 
            \\
        && \circ \bigcirc_{\ell=1}^{j-1} Y_\ell
            &\qquad&
                && \circ \bigcirc_{\ell=1}^{j-1} X_{\ell}
            \\
        && \circ \ T_i
            &\qquad&
                && \circ \ X_j
            \\
        && \circ \bigcirc_{\ell=j+1}^q Z_\ell
            &\qquad&
                && \circ \bigcirc_{\ell=j+1}^{q} X_{\ell}
            \\
        && \circ \ (0,6ph) \circ (0,18)
            &\qquad&
                && \circ \ (0,3h) \circ (0,18)
    \end{array}
\]

It suffices to show that $\frechet(Y_\ell, X_\ell) \leq 1$, $\frechet(T_i,X_j) \leq 1$ and $\frechet( Z_\ell,X_\ell) \leq 1.001$, for a sufficiently small choice of $h$. This is equivalent to showing that $\frechet(R,W_\ell) \leq 1$, $\frechet(T_i,W_j) \leq 1$, and $\frechet(S,W_\ell) \leq 1.001$. We traverse these pairs of trajectories synchronously. For $1 \leq k \leq d$, we have $R_k = 0_A + (18k,0)$, $W_{j,k} = B[j][k]_B + (18k,0)$, $T_{i,k} = A[i][k]_A + (18k,0)$ and $S_k = 0_A + (18k,7ph)$. Since $\frechet(0_A,B[j][k]_B) \leq 1$, we have $\frechet(R_k,W_{j,k}) \leq 1$. Putting this together for $1 \leq k \leq d$, we have $\frechet(R,W_\ell) \leq 1$. Also, we have $\frechet(A[i][k]_A,B[j][k]_B) \leq 1$ for all $1 \leq k \leq d$, since $A[i]$ and $B[k]$ are orthogonal. Therefore, $\frechet(T_{i,k},W_{j,k}) \leq 1$, and putting this together for $1 \leq k \leq d$, we have $\frechet(T_i,W_j) \leq 1$. Finally, since $\frechet(0_A,B[j][k]_B) \leq 1$, we have $\frechet(S_k,W_{j,k}) \leq 1 + 7ph$. We can obtain that $1 + 7ph \leq 1.001$ by setting $h = 0.0001/p$. Putting this together for $1 \leq k \leq d$, we have $\frechet(S,W_\ell) \leq 1.001$. To summarise, we have $\frechet(Y_\ell, X_\ell) \leq 1$, $\frechet(T_i,X_j) \leq 1$ and $\frechet( Z_\ell,X_\ell) \leq 1.001$, so $\frechet(\pi, Q) \leq 1.001$ as required.

We show that if our OV instance $A,B$ is a NO-instance, then $\min_{\pi} \frechet(\pi,Q) \geq 3$. Suppose for the sake of contradiction that $A,B$ is a NO-instance but there exists $\pi \in P$ so that $\frechet(\pi,Q) < 3$. First, note that $\pi$ must start at $(0,-18)$ and end at $(0,18)$ since no other vertices in~$P$ can match to the start and end points of~$Q$. Note that $(0,-18) \in U$ and $(0,18) \in V$, and that any path~$\pi$ from~$U$ to~$V$ must pass through one of the curves $T_i$. Without loss of generality, let $T_i$ be a subcurve of~$\pi$. Consider the point $(22,3ih) \in T_i \subset \pi$. This point must match to some point in~$Q$ that is not on the edges $(0,-18) \circ (0,0)$ or $(0,3h) \circ (0,18)$. Therefore, there exists some $j$ so that $(22,3ih) \in \pi$ matches to a point on $X_j$. The point $(22,3ih)$ cannot match to any of the edges $(0,0) \circ (18,0)$, $(18(d+1),2h) \circ (36d,0)$, $(36d,0) \circ (0,3h)$ or $(0,3h) \circ (0,0)$ on $X_j$. Therefore, the point $(22,3ih) \in T_i$ that matches to a point on $W_j$. As the path $W_j \subset Q$ is traversed, the path $\pi$ must continue to traverse the path along $T_i$, since $T_i$ is an isolated path that only connects to the rest of~$P$ at its endpoints. Therefore, $T_i$ and $W_j$ are traversed simultaneously. Specifically, the subcurves~$T_{i,k} \subset T_i$ and~$W_{j,k} \subset W_j$ are traversed simultaneously, since no points on~$T_{i,k}$ can match to points on~$W_{j,k'}$ for all~$k \neq k'$. This implies $\frechet(\pi,Q) \geq \frechet(T_{i,k},W_{j,k})$ for all $1 \leq k \leq d$. 

Finally, we use the fact that $A,B$ is a NO-instance to show that $\frechet(T_{i,k},W_{j,k}) = 3$ for some $1 \leq k \leq d$. Since $A,B$ is a NO-instance, $A[i]$ and $B[j]$ are not orthogonal. Therefore, there exists a $k$ so that $A[i][k] = B[j][k] = 1$. Therefore, $\frechet(A[i][k]_A,B[j][k]_B) = 3$. Since $T_{i,k} = A[i][k]_A + (18k,3ih)$ and $W_{j,k} = B[j][k]_B + (18k,0)$, we have $\frechet(A[i][k]_A,B[j][k]_B) = 3 + 3ih \geq 3$. Therefore, $\frechet(\pi,Q) \geq \frechet(T_{i,k},W_{j,k}) > 3$, contradicting the fact that $\frechet(\pi,Q) < 3$. Therefore, if our OV instance $A,B$ is a NO-instance, then $\min_{\pi} \frechet(\pi,Q) \geq 3$ as required.

To summarise, we can decide if $A,B$ is a YES-instance or a NO-instance by deciding whether $\min_{\pi} \frechet(\pi,Q) \leq 1.001$ or $\min_{\pi} \frechet(\pi,Q) \geq 3$. Therefore, any $2.999$-approximation with running time $O((pq)^{1-\delta})$ yields an algorithm for OV with running time $O(d^2(pq)^{1-\delta})$, where~$p = |A|$ and~$q = |B|$. This contradicts OVH and SETH.
\end{proof}

Finally, we combine the ideas in Lemma~\ref{lemma:frechet_distance_query_lower_bound} and~\ref{lemma:map_matching_lower_bound} to obtain the main theorem of the section. The theorem essentially states that for geometric planar graphs, preprocessing does not help for map matching. In particular, we show that even with polynomial preprocessing time on the graph, one cannot obtain a truly subquadratic query time for the map matching problem in Problem~\ref{problem:mapmatchingqueriesproblem}. Similar ideas were used in Driemel and Psarros~\cite{DBLP:journals/corr/driemelpsarros} and Bringmann et al.~\cite{DBLP:conf/soda/BringmannDNP22} for the Approximate Nearest Neighbour Fr\'echet query problem.

\lowerboundmapmatchingquery*

\begin{proof}
Suppose for the sake of contradiction that there exists positive constants $\alpha$ and $\delta$ so that one can construct a data structure in $O(n^\alpha)$ preprocessing time to answer $2.999$-approximate map matching queries with a query time of $O((pq)^{1-\delta})$. 

Suppose $q = \Theta(p^\gamma)$ for some $\gamma > 0$. We take two cases. In the first case, $\gamma \geq 2\alpha$. Given a geometric planar graph of complexity~$p$ and a trajectory of complexity~$q$, we can preprocess the graph in $O(p^\alpha)$ time, and query a $2.999$-approximation of $\min_{\pi} \frechet(\pi,Q)$ in $O((pq)^{1-\delta})$ time. But $O(p^\alpha) = O(q^{1/2})$, so the overall running time is $O((pq)^{1-\delta})$. This contradicts Lemma~\ref{lemma:map_matching_lower_bound}. 

In the second case, $\gamma \leq 2 \alpha$. We are given an OVH instance $A,B$ where $|A| = n$ and $|B| = m$. Since OVH holds for any polynomial restrictions of $n$ and $m$, we may assume that $m = \Theta(n^{2\alpha})$. Partition the set~$B$ into subsets $B_1, \ldots, B_K$ so that $|B_i| = \Theta(|A|^\gamma)$, for all $1 \leq i \leq K$, and $K = O(|A|^{2\alpha - \gamma})$. Note that $A \times B$ contains a pair of orthogonal vectors if and only if there exists $1 \leq i \leq K$ so that $A \times B_i$ contains a pair of orthogonal vectors. Given the OV instance $A,B_i$, we use  Lemma~\ref{lemma:map_matching_lower_bound} to construct a geometric planar graph~$P$ and a trajectories $Q_i$ so that $p = |P| = O(d |A|) = O(dn)$ and $q = |Q_i| = O(d |B_i|) = O(dn^\gamma)$. Moreover, by Lemma~\ref{lemma:map_matching_lower_bound}, if $(A,B_i)$ is a YES-instance, then $\min_{\pi}(\pi,Q_i) \leq 1.001$, whereas if $(A,B_i)$ is a NO-instance, then $\min_{\pi}(\pi,Q_i) \geq 3$. Note that $q = \Theta(p^\gamma)$.

Therefore, to decide if $A,B$ is a YES-instance or a NO-instance, it suffices to query a $2.999$-approximation of $\min_{\pi}(\pi,Q_i)$ for all $1 \leq i \leq K$. Recall that $m = \Theta(n^{2 \alpha})$. We preprocess the graph~$P$ in $O((dn)^\alpha) = O(d^\alpha m^{1/2})$ time. We answer all $K$ queries in time

\[
    \begin{array}{rcl}
        O(\sum_{i=1}^K(dn)^{1-\delta}(dn^\gamma)^{1-\delta})
        &=& O(Kd^2 n^{1-\delta+\gamma}) \\
        &=& O(d^2 n^{(1-\delta+\gamma + 2\alpha - \gamma}) \\
        &=& O(d^2 n^{2 \alpha + 1 - \delta}) \\
        &=& O(d^2 n^{(2 \alpha + 1)(1 - \delta/(1+2\alpha)}) \\
        &=& O(d^2 (mn)^{1 - \delta/(1+2\alpha)}).
    \end{array}
\]
Putting this together, we yield an algorithm for OVH with running time 
\[
    O(d^\alpha m^{1/2} + d^2 (mn)^{1 - \frac \delta {1+2\alpha}}),
\]
where $\alpha$ and $\delta$ are constants. This contradicts OVH under the polynomial restriction $m = \Theta(n^{2\alpha})$, thereby contradicting SETH.
\end{proof}

\section{Conclusion}

We showed that for $c$-packed graphs, one can construct a data structure of near-linear size, so that map matching queries can be answered in time near-linear in terms of the query complexity, and polylogarithmic in terms of the graph complexity. We showed that for geometric planar graphs, there is no data structure for answering map matching queries in truly subquadratic time, unless SETH fails. 

Our map matching queries return the minimum Fr\'echet distance between the query trajectory and any path in an undirected graph. The data structure can be modified for directed graphs and for matched paths that start and end along edges of the graph. We can also modify the data structure to return the length of the minimum Fr\'echet distance path. More generally, one can modify the data structure to return $\sum_{e \in \pi} f(e)$ for any function~$f$, where~$\pi$ is the path with approximate minimum Fr\'echet distance. One application of this is fare estimation for ride-sharing services. 

Our data structures return the minimum Fr\'echet distance of the matched path. One can modify our data structure to retrieve the path that attains the minimum Fr\'echet distance, however, the space requirement would increase to quadratic. An open problem is whether one can obtain a map matching data structure that retrieves the matched path, and uses subquadratic space. Another open problem is whether one can make $\varepsilon$ chooseable at query time, rather than at preprocessing time.

Yet another direction for future work is to improve the preprocessing, size, and query time of the data structure. Can one improve the preprocessing time to subquadratic? Can one reduce the dependencies on $c$, $\varepsilon^{-1}$, $\log q$ and $\log p$? For example, can one improve the query time by avoiding parametric search? Avoiding parametric search would also make the algorithm more likely to be implementable in practice. 

Another practical consideration is verifying whether real-world road networks are indeed $c$-packed. Since these road networks contain upwards of a million edges~\cite{DBLP:conf/alenex/ChenDGNW11}, a faster implementation for computing the $c$-packedness value of a graph~\cite{DBLP:conf/isaac/GudmundssonSW20} would be required. If real-world road networks are not $c$-packed, an interesting direction for future work would be to consider other realistic input models, such as $\phi$-low-density, which have small values of $\phi$ even on large road networks~\cite{chen2008faster}.

Finally, two open problems are proposed in Section~\ref{sec:lower_bounds}. Can one modify the lower bound of Buchin~et~al.~\cite{DBLP:conf/soda/BuchinOS19} to rule out approximation ratios between 1.001 and 3 for preprocessing a trajectory to answer Fr\'echet distance queries in truly subquadratic time? Can one extend the lower bounds to rule out efficient data structures for other Fr\'echet distance queries, for example, range searching queries?

\bibliographystyle{plain}
\bibliography{main.bib}

\end{document}